\documentclass[a4paper,UKenglish,cleveref, autoref, thm-restate,numberwithinsect]{lipics-v2021}
\usepackage{amsthm}
\usepackage{amsmath}
\usepackage{amssymb}
\usepackage{commath}
\usepackage{mathtools}
\usepackage{csquotes}
\usepackage{bm}
\usepackage{tikz}
\usepackage{pgfplots}
\usetikzlibrary{shapes,patterns,arrows,graphs,decorations.pathreplacing,calc,backgrounds}
\usepackage{xcolor}
\usepackage{graphicx,color}
\usepackage{xspace}
\usepackage{array}
\usepackage{listings}
\usepackage{cite} %
\usepackage[vlined,linesnumbered,ruled]{algorithm2e}

\DeclareMathOperator{\EX}{\mathbb{E}}
\DeclareMathOperator{\E}{\mathbb{E}}

\graphicspath{{figs/}}

\newcommand{\ALG}{\ensuremath{\mathrm{ALG}}}
\newcommand{\OPT}{\ensuremath{\mathrm{OPT}}}

\newcommand{\bestVC}{\textsc{BestVC}\xspace}
\newcommand{\threshold}{\textsc{Threshold}\xspace}

\newcommand{\intervalcolor}[5]{
  \draw[#5] (#2, #4) node[#5,anchor=east]{#1} -- (#3, #4);
  \draw[#5] (#2, #4-0.1) -- (#2, #4+0.1);
  \draw[#5] (#3, #4-0.1) -- (#3, #4+0.1);
}
\newcommand{\interval}[4]{
  \intervalcolor{#1}{#2}{#3}{#4}{black}
}

\tikzset{vertex/.style={circle, draw, fill,inner sep=0pt, minimum width=4pt}}

\title{Orienting (hyper)graphs under explorable stochastic uncertainty}

\titlerunning{Orienting (hyper)graphs under explorable stochastic uncertainty} 

\author{Evripidis Bampis}{Sorbonne Université, CNRS, LIP6, France \and \url{https://www.lip6.fr/Evripidis.Bampis/} }{evripidis.bampis@lip6.fr}{https://orcid.org/0000-0002-4498-3040}{Supported by the grant ANR-19-CE48-0016 from the French National Research Agency (ANR)}
\author{Christoph Dürr}{Sorbonne Université, CNRS, LIP6, France \and \url{https://www.lip6.fr/Christoph.Durr/} }{christoph.durr@lip6.fr}{https://orcid.org/0000-0001-8103-5333}{Supported by the grant ANR-19-CE48-0016 from the French National Research Agency (ANR)}
\author{Thomas Erlebach}{School of Informatics, University of Leicester, United Kingdom \and \url{https://www.cs.le.ac.uk/people/te17/} }{te17@leicester.ac.uk}{https://orcid.org/0000-0002-4470-5868}{Supported by EPSRC grant EP/S033483/1.}
\author{Murilo Santos de Lima}{School of Informatics, University of Leicester, United Kingdom \and \url{https://www.ime.usp.br/~mslima/} }{mslima@ic.unicamp.br}{https://orcid.org/0000-0002-2297-811X}{Funded by EPSRC grant EP/S033483/1.}
\author{Nicole Megow}{Department of Mathematics and Computer Science, University of Bremen, Germany \and \url{https://www.uni-bremen.de/en/cslog/nmegow} }{nicole.megow@uni-bremen.de}{https://orcid.org/0000-0002-3531-7644}{Supported by the German Science Foundation (DFG) under contract ME~3825/1.}
\author{Jens Schlöter}{Department of Mathematics and Computer Science, University of Bremen, Germany \and \url{https://www.uni-bremen.de/en/cslog/team/jens-schloeter} }{jschloet@uni-bremen.de}{https://orcid.org/0000-0003-0555-4806}{Funded by the German Science Foundation (DFG) under contract ME~3825/1.}

\authorrunning{E. Bampis, C. Dürr, T. Erlebach, M.\,S. de Lima, N. Megow, and J. Schlöter} 

\Copyright{Evripidis Bampis, Christoph Dürr, Thomas Erlebach, Murilo S. de Lima, Nicole Megow, and Jens Schlöter} 

\ccsdesc[500]{Theory of Computation~Design and analysis of algorithms}
\ccsdesc[200]{Mathematics of computing~Discrete mathematics}

\keywords{explorable uncertainty, queries, stochastic optimization, graph orientation, selection problems} 

\acknowledgements{}

\nolinenumbers 

\hideLIPIcs  

\EventEditors{John Q. Open and Joan R. Access}
\EventNoEds{2}
\EventLongTitle{42nd Conference on Very Important Topics (CVIT 2016)}
\EventShortTitle{CVIT 2016}
\EventAcronym{CVIT}
\EventYear{2016}
\EventDate{December 24--27, 2016}
\EventLocation{Little Whinging, United Kingdom}
\EventLogo{}
\SeriesVolume{42}
\ArticleNo{23}

\begin{document}

\maketitle

\begin{abstract}
Given a hypergraph with uncertain node weights
following known probability distributions, we study
the problem of querying as few nodes as possible until
the identity of a node with minimum weight can be determined
for each hyperedge. Querying a node has a cost and reveals
the precise weight of the node, drawn from the given probability
distribution. Using competitive analysis, we compare
the expected query cost of an algorithm with the
expected cost of an optimal query set for the
given instance.
For the general case, we give a polynomial-time
$f(\alpha)$-competitive algorithm, where
$f(\alpha)\in [1.618+\epsilon,2]$ depends on
the approximation ratio $\alpha$ for an underlying
vertex cover problem.
We also
show that no algorithm using a similar approach can be better
than $1.5$-competitive.
Furthermore, we give polynomial-time $4/3$-competitive
algorithms for bipartite graphs with arbitrary query costs
and for hypergraphs with a single
hyperedge and uniform query costs, with matching lower bounds.
\end{abstract}

\section{Introduction}
\label{sec:intro}
The research area of explorable uncertainty is concerned with scenarios where
parts of the input are initially uncertain, but the precise weight (or value)
of an input item can be obtained via a \emph{query}.
For example, an uncertain weight may be represented as an interval
that is guaranteed to contain the precise weight, but only a query
can reveal the precise weight.
Adaptive algorithms make queries one by one until they have gathered sufficient
information to solve the given problem.
The goal is to make as few queries as possible.
In most of the previous work on explorable uncertainty, an adversarial model
has been studied where an adversary
determines the precise weights of the
uncertain elements in such a way that the performance of the algorithm,
compared to the optimal query set, is as bad as possible. While
this model provides worst-case guarantees that hold for every possible
instance, it is also very pessimistic because the adversary is free
to choose the precise weights arbitrarily. In realistic scenarios,
one may often have some information about where in the given
interval the precise weight of an uncertain element is likely to lie.
This information can be represented as a probability distribution
and exploited in order to achieve better performance guarantees.

In this paper, we study the following problem under stochastic
uncertainty: Given a family of (not necessarily disjoint) subsets
of a set of uncertain elements, determine the element with minimum
precise weight in each set, using queries of minimum total cost.
Note that we do not necessarily need to obtain the precise minimum weight.
We phrase the problem in the language of hypergraphs, where
each uncertain element corresponds to a node and each set
corresponds to a hyperedge.
We call this the \emph{hypergraph orientation} problem, as we can think of orienting each hyperedge towards its minimum-weight vertex.
Each node $v\in V$ of a hypergraph $H=(V,E)$
is associated with a known continuous probability distribution\footnote{%
We assume the distribution is given in such a way that ${\mathbb P}[w_v \in (a,b)]$ can be computed in {polynomial} time
for every $v\in V, a,b\in\mathbb R$. For all our algorithms
it suffices to be given a probability matrix: rows correspond to vertices $v$,
columns to elementary intervals $(t_i,t_{i+1})$, and entries to ${\mathbb
P}[w_v \in (t_i, t_{i+1})]$, where $t_1,\ldots,t_{2|V|}$ represent the
sorted elements of~$\{\ell_v, r_v | v \in V\}$.} $d_v$ over an interval $I_v=(\ell_v,r_v)$
and has query cost~$c_v$.
The precise weight of $v$ is drawn independently from~$d_v$ and
denoted by~$w_v$.
We assume that $I_v$ is the minimal interval
that contains the support of $d_v$, i.e., $\ell_v$ is the largest value
satisfying ${\mathbb P}[w_v\le \ell_v]=0$ and $r_v$ is the smallest value
satisfying ${\mathbb P}[w_v\ge r_v]=0$.
For~$S\subseteq V$, we define $c(S)=\sum_{v\in S} c_v$.
An algorithm
can sequentially make queries to
vertices
to learn their weights,
until it has enough information to identify the minimum-weight vertex of each
hyperedge.
A query of $v$ reveals its precise
{weight} $w_v$, which is drawn independently from~$d_v$.
If all vertices have the same query cost, we say that
the query costs are uniform and assume w.l.o.g.\ that
$c_v=1$ for all $v\in V$.
Otherwise, we speak of arbitrary query costs.
The objective of an algorithm is to minimize the expected cost of the queries
it makes.

We also consider the special case where we are given a graph
$G=(V,E)$ instead of a hypergraph $H=(V,E)$, called the
\emph{graph orientation} problem. 

{As an example consider a multi-national medical company that needs a certain product (say, chemical ingredient, medicine or vaccine) for its operation in each country. The particular products that are available in each country are different due to different approval mechanisms. The task is to find the best product for each country, that is, the best among the approved ones. The quality itself is independent of the country and can be determined by extensive tests in a lab (queries). The set of products available in one country corresponds to a hyperedge, and the problem of identifying the best product in each country is the hypergraph orientation~problem.}

\subparagraph{Our contribution.} 
Our main result (Section~\ref{sec:threshold}) is an algorithm for the graph orientation problem with competitive ratio $\frac{1}{2}(\alpha + \sqrt{8 - \alpha(4 - \alpha)})$, assuming we have an $\alpha$-approximation for the vertex cover problem (which we need to solve on an induced subgraph of the given graph). This factor is always between $\phi \approx 1.618$ (for $\alpha = 1$), and $2$ (for $\alpha = 2$).
{We show that, for the special cases of directing $\mathcal{O}(\log \ |V|)$ hyperedges and sorting $\mathcal{O}(1)$ sets, the algorithm can be applied with $\alpha = 1$ in polynomial running time.}
The algorithm has a preprocessing phase in two steps.
First, we compute the probability that a vertex is \emph{mandatory}, i.e., that it is part of any feasible solution, and we query all vertices with probability over a certain threshold.
The second step uses a LP relaxation of the vertex cover problem to select some further vertices to query.
Next, we compute an $\alpha$-approximation of the vertex cover on a subgraph induced by the preprocessing, and we query the vertices in the given solution.
The algorithm finishes with a postprocessing that only queries mandatory intervals.
For the analysis, we show two main facts: (1) the expected optimal solution can be bounded by the expected optimal solutions for the subproblems induced by a partition of the vertices; (2) for the subproblem on which we compute a vertex cover, the expected optimal solution can be bounded by applying the K\H{o}nig-Egerváry theorem~\cite{SchrijverBook} on a particular bipartite graph, in case of uniform costs.
When given arbitrary query costs, 
we
utilize a technique of \emph{splitting} the vertices in order to obtain a collection of disjoint stars with obvious vertex covers that imply a bound on the expected~optimum.

We further show how to generalize the algorithm to hypergraphs.
Unfortunately in this case it is \#P-hard to compute the probability of a vertex being mandatory, but we can approximate it by sampling.
This yields a randomized algorithm that attains, with high probability, a competitive ratio arbitrarily close to the expression given above for graphs.
Here, we need to solve the vertex cover problem on an induced subgraph of an auxiliary graph that contains, for
each hyperedge of the given hypergraph, all edges between the node with the leftmost interval and the nodes whose intervals intersect that interval.

We also consider a natural 
alternative algorithm (Section~\ref{sec:VC-alg}) that starts with a particular vertex cover solution followed by adaptively querying remaining vertices. We prove a competitive ratio of $4/3$ on special cases, namely, for bipartite graphs with arbitrary cost and for a single hyperedge with uniform costs, and complement this by matching lower bounds.

\subparagraph{Related work.}
Graph orientation problems are fundamental in the area of graph theory and combinatorial optimization. In general, graph orientation refers to the task of giving an orientation to edges in an undirected graph such that some given requirement is met. Different types of requirements have been investigated. While Robbins~\cite{robbins1939} initiated research on connectivity and reachability requirements already in the 1930s, most work is concerned with degree-constraints; cf.~overviews given by Schrijver~\cite[Chap.~61]{SchrijverBook} and Frank~\cite[Chap.~9]{FrankBook}.

Our requirement, orienting each edge towards its node with minimum weight, becomes 
challenging when there is 
uncertainty in the node weights. While there are different ways of modeling uncertainty in the input data, {the model of explorable uncertainty was introduced by Kahan~\cite{kahan91}.}
He considers the task of identifying the minimum element in a set of 
uncertainty intervals, which is equivalent to  orienting a single hyperedge. Unlike in our model, 
no distributional information is known, and an adversary can choose weights in a worst-case manner from the intervals. Kahan~\cite{kahan91} shows that querying the intervals in order of non-decreasing left endpoints requires at most one more query than the optimal query set, thus giving a competitive ratio of~$2$. {Further, he shows that this is best possible in the adversarial model}.

Subsequent work 
addresses finding the~$k$-th smallest value in a set of uncertainty intervals~\cite{gupta16queryselection,feder03medianqueries}, 
caching problems 
\cite{olston2000queries}, computing a function value~\cite{khanna01queries}, 
and classical combinatorial optimization problems, such as shortest path~\cite{feder07pathsqueires},  knapsack~\cite{goerigk15knapsackqueries}, scheduling problems~\cite{DurrEMM20,arantes18schedulingqueries,albersE2020}, minimum spanning tree and matroids~\cite{erlebach08steiner_uncertainty,erlebach14mstverification,megow17mst,focke20mstexp,MerinoS19}.
Recent work on sorting elements of a single or multiple non-disjoint sets is particularly relevant as it is a special case of the graph orientation problem~{\cite{ErlebachHLMS-arxiv2020,halldorssonL21sortingfull}}. For sorting a single set in the adversarial explorable uncertainty model, there is a $2$-competitive algorithm and it is  best possible, even for arbitrary query costs~\cite{halldorssonL21sortingfull}. The competitive ratio can be improved to~$1.5$ for uniform query cost by using randomization~\cite{halldorssonL21sortingfull}. Algorithms with limited adaptivity have been proposed in~\cite{ErlebachHL2021}.

Although the adversarial model is arguably pessimistic and real-world applications often come with some distributional information, surprisingly little is known on stochastic variants of explorable uncertainty.
The only previous work we are aware of is by Chaplick et al.~\cite{chaplick20stochasticLATIN}, in which they studied stochastic
uncertainty for the problem of sorting a given set of uncertain elements,
and for the problem of determining the minimum element in a given set
of uncertain elements. They showed that the optimal decision tree
(i.e., an algorithm that minimizes the expected query cost among
all algorithms) for a given instance of the sorting problem can be
computed in polynomial time. For the minimum problem, they
leave open whether an optimal decision tree can be determined
in polynomial time, but give a $1.5$-competitive algorithm
and an algorithm that guarantees a bound slightly smaller than
$1.5$ on the expectation of the ratio between the query cost
of the algorithm and the optimal query cost.
The problem of scheduling with testing~\cite{LeviMS19} is also in the spirit of stochastic explorable uncertainty but less~relevant~here.

{There are many other} stochastic problems that take exploration cost into account. Some of the earliest work has studied multi-armed bandits~\cite{Thompson33,BubeckC12,GittinsGW11-book} and Weitzman's Pandora's box problem \cite{Weitzman1979}, which are prime examples for analyzing the tradeoff between the cost for exploration and the benefit from exploiting gained information. More recently, query-variants of combinatorial optimization problems received 
some attention, 
in general~\cite{singla2018price,gupta2019markovian}, and for specific problems such as stochastic knapsack~\cite{DeanGV08,Ma18}, orienteering~\cite{GuptaKNR15,BansalN15}, matching~\cite{ChenIKMR09,BansalGLMNR12,BlumDHPSS20,BehnezhadFHR19,AssadiKL19}, and probing problems \cite{AdamczykSW16,GuptaN13,GuptaNS16}. Typically such work employs a \emph{query-commit} model, 
{meaning} 
that queried elements must be part of the solution, {or solution elements are required to be queried}. These are quite strong requirements that lead to a different flavor of the cost-benefit tradeoff. 

Research involving queries to identify particular graph structures or elements, or queries to verify certain properties, can be found in various flavors. A well-studied problem class is \emph{property testing}~\cite{Goldreich2017}, and there are many more, see e.g.,~\cite{Mazzawi2010,Beame2018,Rubinstein2018,Chen2020,Assadi2021,Nisan2021}. Without describing such problems in detail, we emphasize a fundamental difference to our work. Typically, in these query models, the bounds on the number of queries made by an algorithm are \emph{absolute numbers}, i.e., given as a function of the input size, but independent of the input graph itself 
and without any comparison to the minimum number of queries needed for the given graph.

\section{Definitions and preliminary results}
\label{sec:preliminaries}

The hypergraph orientation problem and the graph orientation
problem have already been defined in Section~\ref{sec:intro}.
In this section we first give additional definitions and
discuss how we measure the performance of an algorithm.
Then we introduce the concept of mandatory vertices and show how
the probability for a vertex to be mandatory can be computed or
at least approximated efficiently.
We also give a lower bound showing that no
algorithm can achieve competitive ratio better than~$\frac{4}{3}$.
Next, based on the concept of witness sets,
we define the vertex cover instance associated with an instance
of our problem and define a class of vertex cover-based algorithms,
which includes all the algorithms we propose in this paper.
Finally, we characterize the optimal query set for each realization
and give lower bounds on the expected optimal query cost, which we will
use later in the analysis of our algorithms.

\subparagraph{Definitions.}
To measure the performance of an algorithm, we compare the expected
cost of the queries it makes to
the expected optimal query cost. Formally, given a realization of the values, we call \emph{feasible query set}
a set of vertices to be queried that permits one to identify the minimum-weight vertex
in every hyperedge.
Note that a query set is feasible if, for each hyperedge, it either queries the
node $v$ with minimum weight $w_v$ and all other nodes whose intervals contains~$w_v$,
or it does not query the node $v$ with minimum weight but queries all nodes whose interval
overlaps $I_v$, and in addition the precise weights of all those intervals lie to the right
of~$I_v$.
An \emph{optimal query set} is a feasible query set
of minimum query cost.  We denote by $\EX[\OPT]$ the
expected query cost of an optimal query set. Similarly, we denote by
$\EX[A]$ the expected query cost of the query set queried by an algorithm~$A$. The
supremum of $\EX[A] / \EX[\OPT]$, over all instances of
the problem, is called the \emph{competitive ratio} of~$A$.
{Alternatively, one could compare $\EX[A]$ against the cost $\EX[A^*]$ of an optimal adaptive algorithm $A^*$.
However, in explorable uncertainty, it is standard to compare against the optimal query set, and, since $\EX[OPT]$ is a lower bound on $\EX[A^*]$, all our algorithmic results translate to this alternative setting.}

Let $F\in E$ be a hyperedge consisting of vertices $v_1,\ldots,v_k$,
indexed in order of non-decreasing left endpoints of the intervals, i.e., $\ell_{v_1}
\leq \ldots \leq \ell_{v_k}$.  We call $v_1$ the \emph{leftmost} vertex of $F$.
We can assume that $I_{v_1}\cap I_{v_i} \neq \emptyset$ for all $2\le i\le k$,
because otherwise the vertex $v_i$ could be removed from the hyperedge $F$.
For the special case of graphs, this means that we
assume $I_v \cap I_u \neq \emptyset$ for each $\{u,v\} \in E$, since
otherwise we could simply remove the edge.

\subparagraph{Mandatory vertices, probability to be mandatory.}
A vertex $v$ is called \emph{mandatory} if it belongs to every feasible query set for the
given realization.  For example, if for some edge $\{u,v\}$, vertex $u$ has
already been queried and its value $w_u$ belongs to the interval $I_v$, then $v$ is known to be mandatory.
The following lemma was shown in~\cite{ErlebachHLMS-arxiv2020} and fully characterizes mandatory vertices.

\begin{restatable}[]{lemma}{lemmamandatory}
	\label{lemma:mandatory}
	A vertex $v \in V$ is mandatory if and only if there is a hyperedge $F \in E$ with $v \in F$ such that either $(i)$ $v$ is a minimum-weight vertex of $F$ and $w_u \in I_v$ for some $u \in F \setminus \{v\}$, or $(ii)$ $v$ is not a minimum-weight vertex of $F$ and $w_u \in I_v$ for the minimum-weight vertex $u$ of $F$.
\end{restatable}

For a hyperedge $F=\{v_1,\ldots,v_k\}$, where the vertices are again
indexed by non-decreasing left endpoints,
it was shown in~\cite[Section~3]{chaplick20stochasticLATIN} that, if $I_{v_i} \subseteq I_{v_1}$ for some $2 \le i \le k$, then $v_1$ is mandatory for every realization.
Thus, every algorithm can iteratively query all such elements in a preprocessing step, without worsening the competitive ratio. 
In the remainder of the paper, we assume w.l.o.g.~that the instance under consideration is already preprocessed.

Similarly, if a hyperedge contains vertices $u,v$ such that $v$ has not been queried yet and
is the leftmost vertex, while a query of $u$ has revealed that $w_u \in I_v$, then it follows from
Lemma~\ref{lemma:mandatory} that $v$ is mandatory for every realization of the unqueried vertices.
The final stage of our algorithms will consist of querying mandatory vertices that are identified
by this criterion, until the instance is solved.

We denote by $p_v$ the probability that a vertex~$v$ is mandatory.
Querying vertices $v \in V$ that have a high probability $p_v$
is a key element of our algorithms. 
For graphs, $p_v$ is easy to compute as, by Lemma~\ref{lemma:mandatory}, $v$ is mandatory
iff $w_u \in I_v$ for some neighbor vertex $u$. Hence,
$p_v=1 - \prod_{u:\{u,v\}\in E} {\mathbb P}[w_u \not\in I_v]$.
For hypergraphs, however, we
can show that the computation of $p_v$
is \#P-hard, even if all hyperedges have size~$3$ (see Appendix~\ref{app:sharpP}).
Luckily it is not difficult to get a good estimate of the probabilities to be mandatory for hypergraphs using sampling.

\begin{restatable}[]{lemma}{lemsampling}
\label{lem:sampling}
There is a polynomial-time randomized algorithm that,
	given a hypergraph $H=(V,E)$, a vertex $v\in V$, and parameters $\epsilon, \delta \in (0,1)$, produces a value $y$ such that $|y-p_v| \geq \epsilon$ with probability at most $\delta$. Its time complexity is $O(|V| \ln(1/\delta)/ \epsilon^2)$.
\end{restatable}

\subparagraph{General lower bound.}
We have the following lower bound.

\begin{theorem}
\label{thm:generalLB43}
Every algorithm for the graph orientation problem has competitive ratio at least $\frac{4}{3}$, even for uniform query costs and even if no restriction on the {running time} of the algorithm is imposed.
\end{theorem}

\begin{proof}
	\begin{figure}[tb]
		\centering
		\begin{minipage}{0.3\textwidth}
		\centering
		\begin{tikzpicture}[thick, scale=0.5]
			\draw (0, 0) node[vertex, label=west:$x$]{} -- (4, -1) node[vertex, label=east:$z$]{};
			\draw (0, 0) -- (4, 1) node[vertex, label=east:$y$]{};
		\end{tikzpicture}
		\end{minipage}
		\begin{minipage}{0.3\textwidth}
		\centering
		\begin{tikzpicture}[line width = 0.3mm]
			\interval{$x$}{0}{2}{1}
			\interval{$y$}{1}{3}{0.5}
			\interval{$z$}{1}{3}{0}
			
			\draw[dotted] (1, 1.4) -- (1, -0.4);
			\draw[dotted] (2, 1.4) -- (2, -0.4);
			
			\node at (0.5, 1.2) {$1/2$};
			\node at (1.5, 1.2) {$1/2$};
			\node at (1.5, 0.65) {$\epsilon$};
			\node at (2.5, 0.65) {$1 - \epsilon$};
			\node at (1.5, 0.15) {$\epsilon$};
			\node at (2.5, 0.15) {$1 - \epsilon$};
		\end{tikzpicture}
		\end{minipage}
		\begin{minipage}{0.38\textwidth}
			$$
			\begin{array}{lcl}
				p_x & = & 1-(1-\epsilon)^2\\
				p_y & = & 1/2\\
				p_z & = & 1/2
			\end{array}
			$$
		\end{minipage}
	
		\caption{Instance and mandatory probabilities used in the proof of Theorem~\ref{thm:generalLB43}.}
		\label{fig:generalLB43}
	\end{figure}
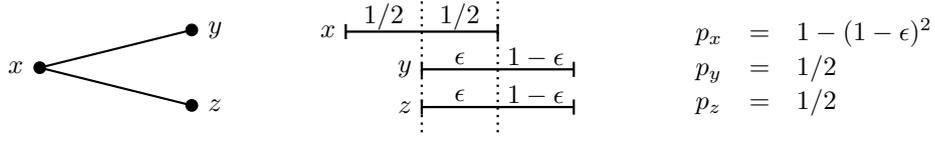%
Consider three vertices $x, y, z$, with $I_x = (0, 2)$ and $I_y = I_z = (1, 3)$, and uniform query costs $c_x = c_y = c_z = 1$.
The only edges are $\{x,y\}$ and $\{x,z\}$.
The probabilities are such that $\mathbb{P}[w_x \in (1, 2)] = \frac{1}{2}$ and $\mathbb{P}[w_y \in (1, 2)] = \mathbb{P}[w_z \in (1, 2)] = \epsilon$, for some $0 < \epsilon \ll \frac{1}{2}$; see Figure~\ref{fig:generalLB43}.
If $w_x\in(0,1]$, which happens with probability $\frac12$,
querying $x$ is enough. If $w_x\in(1,2)$ and
$w_y,w_z\in[2,3)$, which happens with probability
$\frac12 (1-\epsilon)^2$, querying $y$ and $z$ is enough.
Otherwise, all three vertices must be queried.
We have
\begin{displaymath}
 \EX[\OPT] = \frac{1}{2} \cdot 1 + \frac{1}{2} (1-\epsilon)^2 \cdot 2 + \frac{1}{2} \left( 1 - (1 - \epsilon)^2 \right) \cdot 3 = 2 - \frac{(1-\epsilon)^2}{2},
\end{displaymath}
which tends to $\frac{3}{2}$ as $\epsilon$ approaches $0$.
Since $y$ and $z$ are identical and we can assume that an algorithm
always queries first a vertex that it knows to be mandatory (if there
is one), we only have
three possible decision trees to consider:
\begin{enumerate}
 \item First query $x$; if $w_x \in (1, 2)$, then query $y$ and $z$. The expected query cost is $2$.
 \item First query $y$. If $w_y \in (1, 2)$, then query $x$, and query $z$ if $w_x \in (1, 2)$. If $w_y \in [2, 3)$, then query $z$, and query $x$ if $w_z \in (1, 2)$. The expected query cost is $1 + \frac{3}{2} \epsilon + (1 - \epsilon) (1 + \epsilon)$, which tends to $2$ as $\epsilon$ approaches $0$.
 \item First query $y$. Whatever happens, query $x$, then query $z$ if $w_x \in (1, 2)$. The expected query cost is $\frac{5}{2}$, so this is never better than the previous options.
\end{enumerate}
With either choice (even randomized), the competitive ratio tends to at least $\frac{4}{3}$ as $\epsilon\rightarrow 0$.
\end{proof}

This lower bound can be adapted for a single hyperedge $\{x, y, z\}$.
For arbitrary query costs, it works even for a single edge $\{x, y\}$, by taking $c_x = 1$ and $c_y = 2$.

\subparagraph{Witness sets, vertex cover instance, vertex cover-based algorithms.}
Another key concept of our algorithms is to exploit \emph{witness sets}~\cite{bruce05uncertainty,erlebach08steiner_uncertainty}.
A subset $W \subseteq V$ is a \emph{witness set} if $W \cap Q \not= \emptyset$ for all
feasible query sets $Q$.
The following lemma was shown in~\cite{kahan91}.

\begin{lemma}
	\label{lemma:witness_set}
	Let $F = \{v_1, \ldots, v_k\}$ be a hyperedge, and let $v_1$ be the leftmost vertex of $F$. Then $\{v_1,v_i\}$ is a witness set for each $2 \le i \le k$.
\end{lemma}

The lemma implies that one can obtain a $2$-competitive algorithm in the adversarial
model for the hyperedge orientation problem:
For uniform query costs, it suffices to repeatedly query witness sets of size~$2$ until the
instance is solved, by a standard witness set argument~\cite{erlebach08steiner_uncertainty}.
For arbitrary query costs, this approach can be combined with the local ratio technique~\cite{baryehuda04}
to obtain the same competitive ratio (in a similar way as done in~\cite{halldorssonL21sortingfull} for
the sorting problem).
Our goal is to achieve better competitive ratios in the stochastic setting.
Motivated by Lemma~\ref{lemma:witness_set}, we can now define the \emph{vertex cover instance}.

\begin{definition}
	\label{def:vertex_cover_instance}
	Given a hypergraph $H = (V,E)$, the \emph{vertex cover instance} of $H$ is the graph $\bar{G}=(V,\bar{E})$ with
	$\{v,u\} \in \bar{E}$ if and only if there is a hyperedge $F \in E$ such that $v,u \in F$, $v$ is leftmost in $F$ and $I_v \cap I_u \neq \emptyset$.
	For the special case of a graph $G$ instead of a hypergraph $H$, it holds that $\bar{G}=G$.
\end{definition}

Since each edge of the vertex cover instance $\bar{G}$ is a witness set by Lemma~\ref{lemma:witness_set}, we can observe that each feasible query set $Q$ is a vertex cover of $\bar{G}$.
Using the vertex cover instance, we can define a class of algorithms for the hypergraph orientation problem as follows:
An algorithm is \emph{vertex cover-based} if it implements the following pattern:
\begin{enumerate}
	\item Non-adaptively query a vertex cover $VC$ of $\bar{G}$;
	\item Iteratively query mandatory vertices until the minimum-weight vertex of each hyperedge is known: For each hyperedge $F \in E$ for which the minimum weight is still unknown, query the vertices in order of left endpoints until the minimum weight is found.
\end{enumerate}
By definition of the second step, each vertex cover-based algorithm clearly orients each hyperedge.
Furthermore, Lemma~\ref{lemma:mandatory} implies that each vertex queried in the last step is indeed mandatory for all realizations that are consistent with the currently known information, i.e., the weights of the previously queried vertices.
For graphs, this is easy to see, and for hypergraphs, this can be shown as follows:
For a hyperedge $F$ that isn't solved after the first step and has leftmost vertex~$v$ initially, the vertex cover $VC$
has queried $v$ or all other vertices of~$F$. In the latter case, $v$~is the only unqueried
vertex of~$F$ and $I_v$~must contain the precise weight of one of the other vertices, hence $v$
is mandatory. In the former case, the remaining candidates for being the minimum-weight vertex are
(1)~the vertex with leftmost precise weight among those queried in the first step, and
(2)~the unqueried vertices whose intervals contains that precise weight. It is then clear
that the leftmost vertex is mandatory, and querying it either solves the hyperedge or
yields a situation of the same~type.

All the algorithms we propose in this paper are vertex cover-based.
We have the following lower bounds for vertex cover-based algorithms.

\begin{restatable}[]{theorem}{theoremVCbasedLB}
	\label{theorem:VC_based_LB}
	No vertex cover-based algorithm has competitive ratio better than $\frac{3}{2}$ for the hypergraph orientation problem.
	This result holds even in the following special cases:
	\begin{enumerate}
		\item The graph has only a single hyperedge but the query costs are not uniform.
		\item The query costs are uniform and the vertex cover instance $\bar{G}$ is bipartite.
		\item 
		{The instance is a non-bipartite graph orientation instance with uniform query costs.}
	\end{enumerate}
\end{restatable}

We remark that the second step of vertex cover-based algorithms must be adaptive:
In Appendix~\ref{app:strict2stage}, we show that any algorithm consisting of two non-adaptive stages cannot have competitive ratio $o(\log n)$, even for a single hyperedge with $n$ vertices and uniform query costs.

\subparagraph{Bounds on $\EX[\OPT]$.} 
\label{sec:lower-bounds-on-OPT}
Let $\mathcal{R}$ be the set of all possible realizations and let $\OPT(R)$ for $R \in \mathcal{R}$ be the optimal query cost for realization~$R$.
As each feasible query set $Q$ must include a vertex cover of $\bar{G}$,
the minimum weight of a vertex cover of $\bar{G}$ (using the query costs as weights) is a lower bound on the optimal query cost for each realization and, thus, on $\EX[\OPT]$.
This observation in combination with Lemma~\ref{lemma:mandatory} also gives us a way to identify an optimal query set for a fixed realization, by using the knowledge of the exact vertex weights. 
\begin{observation}
	\label{obs:verification}
	For a fixed realization $R$ of an instance of the hypergraph orientation problem, let $M$ be the set of vertices that are mandatory (cf.~Lemma~\ref{lemma:mandatory}), and let $VC_M$ be a minimum-weight vertex cover of $\bar{G}[V\setminus M]$. Then $M \cup VC_M$ is an optimal query set for realization $R$.
\end{observation}

Computing $\OPT(R)$ for a fixed and known realization $R$ is NP-hard~\cite{ErlebachHLMS-arxiv2020}. 
{This extends to the hypergraph orientation problem and the 
computation of $\EX[\OPT]$: We can reduce from the problem of computing $\OPT(R)$ by concentrating the probability mass of all 
intervals onto the weights in realization~$R$. 
The reduction of~\cite{ErlebachHLMS-arxiv2020} in combination with~\cite{chlebik2007} also implies APX-hardness.}

To analyze the performance of our algorithms, we compare the expected cost of the algorithms to the expected cost of the optimal solution.
By Observation~\ref{obs:verification}, $c(M) + c(VC_M)$ is the minimum query cost for a fixed realization $R$, where $M \subseteq V$ is the set of mandatory elements in the realization and $VC_M$ is a minimum-weight vertex cover for the subgraph $\bar{G}[V\setminus M]$ of the vertex cover instance $\bar{G}=(V,\bar{E})$ induced by $V\setminus M$.
Thus, the optimal solution for a fixed realization is completely characterized by the set of mandatory elements in the realization.
Using this, we can characterize $\EX[\OPT]$ as
	$\EX[\OPT] = \sum_{M\subseteq V} p(M) \cdot c(M)  +  \sum_{M\subseteq V}  p(M) \cdot c(VC_M)$,  
where $p(M)$ denotes the probability that $M$ is the set of mandatory elements.
It follows that $\sum_{M\subseteq V} p(M) \cdot c(M) = \sum_{v\in V} p_v \cdot c(v)$, since both terms describe the expected cost for querying mandatory elements, which leads to the following characterization of $\EX[\OPT]$:
\begin{equation*}
\EX[\OPT] =  \sum_{v\in V} p_v \cdot c_v   +  \sum_{M\subseteq V}  p(M) \cdot c(VC_M). 
\end{equation*}

{A key technique for 
our analysis}
is lower bounding $\EX[\OPT]$ by partitioning the optimal solution into subproblems and discarding dependencies between elements in different subproblems. 

\begin{definition}\label{def:partial_opt}
	For a realization $R$ and any subset $S \subseteq V$, let $\OPT_S = \min_{Q \in \mathcal{Q}} c(Q \cap S)$, where $\mathcal{Q}$ is the set of all feasible query sets for realization~$R$.
\end{definition}

\begin{lemma}
	\label{lemma_opt_partitioning}
	Let $S_1,\ldots,S_k$ be a partition of $V$. Then  
	$
	\EX[\OPT] \ge \sum_{i=1}^k \EX[\OPT_{S_i}].
	$
\end{lemma}

\begin{proof}
	We start the proof by characterizing $\EX[\OPT_{S_i}]$ for each $i \in \{1,\ldots,k\}$. 
	Let $R \in \mathcal{R}$ be a realization in which $M$ is the set of mandatory elements. Then $\OPT_{S_i}$ needs to contain all mandatory elements of $S_i$, and resolve all remaining dependencies between vertices of $S_i$, i.e., query a minimum-weight vertex cover $VC^{S_i}_{M}$ for the subgraph $\bar{G}[S_i\setminus M]$.
	Thus, it follows	
	\begin{equation}\label{eq_opt_partitioning_1}
		\EX[\OPT_{S_i}] = \sum_{v\in S_i} p_v \cdot c_v + \sum_{M \subseteq V} p(M) \cdot c(VC^{S_i}_{M}).
	\end{equation}
	By summing Equation~\eqref{eq_opt_partitioning_1} over all $i \in \{1,\ldots,k\}$, we obtain the lemma:
	\begin{align*}
		\sum_{i=1}^k \EX[\OPT_{S_i}] &= \sum_{i=1}^k \left(\sum_{v\in S_i} p_v \cdot c_v + \sum_{M \subseteq V} p(M) \cdot c(VC^{S_i}_{M})\right)\\
		&= \sum_{v \in V} p_v \cdot c_v + \sum_{M \subseteq V} p(M) \cdot \left(\sum_{i=1}^k c(VC_M^{S_i})\right)\\
		&\le \sum_{v \in V} p_v \cdot c_v + \sum_{M \subseteq V} p(M) \cdot c(VC_M) = \EX[\OPT],
	\end{align*}
	where the second equality follows from $S_1,\ldots,S_k$ being a partition. 
	The inequality follows from $\sum_{i=1}^k c(VC_M^{S_i})$ being the cost of a minimum weighted vertex cover for a subgraph of $G[V\setminus M]$, while $c(VC_M)$ is the minimum cost for a vertex cover of the whole graph.
\end{proof}

For the case of arbitrary query costs, we will sometimes
need to partition $V$ in such a way that a vertex $v$ can
be split into fractions that are in different parts of
the partition. We view each fraction as a copy of $v$,
and the split is done in such a way that the query costs
of all copies of $v$ add up to~$c_v$. {Further},
the probability distribution for being mandatory in the
resulting instance is such that either all copies of $v$
are mandatory or none of them is, and the former happens
with probability~$p_v$.
(A 
detailed discussion of this process can be found
in Appendix~\ref{app:vertexSplit}.)
We refer to the application
of this operation to a vertex as a \emph{vertex split}
and note that it can be applied repeatedly.

\begin{restatable}[]{observation}{obsvertexsplitting}
\label{obs:vertex_splitting}
	Let $\OPT'$ be the optimal solution for an instance that is created by iteratively executing vertex splits. Then, $\EX[\OPT'] = \EX[\OPT]$. Furthermore, Lemma~\ref{lemma_opt_partitioning} also applies to $\EX[\OPT']$ and the modified instance.
\end{restatable}

\section{A threshold algorithm for orienting hypergraphs}
\label{sec:threshold}

We present an algorithm for orienting graphs and its generalization
to hypergraphs.


\subsection{Orienting graphs}
\label{sec:threshold-algorithm}

{We consider the graph orientation problem.} 
{As a subproblem,}
we solve a 
vertex cover problem. 
{This problem is NP-hard and $2$-approximation algorithms are known \cite{yannakakis_edge_1980}.} For several special graph classes, there are improved algorithms \cite{halperin2002improved}. Using an $\alpha$-approximation 
as a black box, we give a competitive ratio 
between $\phi \approx 1.618$ ($\alpha = 1$) and $2$ ($\alpha = 2$)
as a function depending on $\alpha$.
{Appendix~\ref{app:thres} contains a figure showing the competitive ratio in dependence on $\alpha$.}

\begin{algorithm}[htb]
	\KwIn{Instance $G=(V,E)$, $p_v$ for each $v \in V$, parameter $d \in [0,1]$, \\\phantom{\textbf{Input:\:}}and an $\alpha$-approximation black box for the vertex cover problem}
	Let $M = \{v \in V \mid p_v \ge d\}$\label{line:threshold_addM}\;
	Solve~\eqref{eq:threshold:LP} for $G[V \setminus M]$ and let $x^*$ be an optimal basic feasible solution\label{Line:threshold_LP}\;
	Let $V_1 = \{v \in V \mid x^*_v = 1\}$ and similarly $V_{1/2},V_0$ \label{line:threshold_defV}\;
	Use the $\alpha$-approximation black box to approximate a vertex cover $VC'$ for $G[V_{1/2}]$\label{line:threshold_solvevc}\;
	Query $Q = M \cup V_1 \cup VC'$\label{line:threshold_stage1}\tcc*[l]{$Q$ is a vertex cover}
	Query the mandatory elements of $V \setminus Q$ \label{line:threshold_stage2}\; 
	\caption{\threshold}
	\label{alg:threshold}
\end{algorithm}

Algorithm~\ref{alg:threshold} is parameterized by a threshold $d\in [0,1]$, which is optimized depending on the approximation ratio $\alpha$ of the chosen vertex cover procedure.
The algorithm executes a preprocessing of the vertex cover instance by using the following classical LP relaxation, {for which each optimal basic feasible solution is half-integral}~\cite{nemhauser1975}: 
\begin{equation}\tag{LP}\label{eq:threshold:LP}
\begin{array}{lll}
\min &\sum_{v\in V} c_v\cdot x_v\\
\text{s.t. }& x_v + x_u \ge 1 & \forall \{u,v\} \in E\\
& x_v \ge 0& \forall v \in V
\end{array}
\end{equation}

\begin{restatable}[]{theorem}{theoremThreshold}
	\label{thm:threshold}
	Given an $\alpha$-approximation with $1\le\alpha\le 2$ for the vertex cover problem
	(on the induced subgraph $G[V_{1/2}]$, see Line~\ref{line:threshold_solvevc}),
	\emph{\threshold} with parameter $d$ achieves a competitive ratio of $\max\{\frac{1}{d}, \alpha + (2-\alpha) \cdot d \}$ {for the graph orientation problem.}
	Optimizing~$d$ yields a competitive ratio of $\frac{1}{2}(\alpha + \sqrt{8 - \alpha(4 - \alpha)})$.
\end{restatable}

\begin{proof}
Here, we show the result for uniform query costs. 
The generalization to arbitrary query costs requires an additional technical step involving vertex splitting and is discussed in Appendix~\ref{app:threshold_arbitrary_costs}. 
Since $Q$ is a vertex cover for $G$, querying it in Line~\ref{line:threshold_stage1} and 
resolving all remaining dependencies in Line~\ref{line:threshold_stage2} clearly solves the graph orientation~problem.
Note that $V\setminus Q$ is an independent set in~$G$, and thus the nodes in $V\setminus Q$
can only be made mandatory by the results of the queries to~$Q$. Hence, it is known after
Line~\ref{line:threshold_stage1} which nodes in $V\setminus Q$ are mandatory, and they
can be queried in Line~\ref{line:threshold_stage2} in arbitrary order (or in parallel).

	We continue by showing the competitive ratio of $\max\{\frac{1}{d}, \alpha + (2-\alpha) \cdot d \}$.
	Algebraic transformations show that the  optimal choice for the threshold is $d(\alpha)= 2/(\alpha+\sqrt{8-\alpha(4-\alpha)})$. The desired competitive ratio for \threshold with $d=d(\alpha)$ follows.

	The algorithm queries set $Q$ and all other vertices only if they are mandatory, hence
	\begin{equation}\label{eq_thres_alg}
			\EX[ALG] = |Q| + \sum_{v\in V\setminus Q} p_v 
					= |M| + |V_1| + |VC'| + \sum_{v \in V_0} p_v + \sum_{v \in V_{1/2} \setminus VC'} p_v.
	\end{equation}
	The expected optimal cost can be lower bounded by partitioning and Lemma~\ref{lemma_opt_partitioning}:
	\begin{equation}
	\EX[\OPT] \ge  \EX[\OPT_{M}] + \EX[\OPT_{V_1 \cup V_0}] + \EX[\OPT_{V_{1/2}}].
	\label{eq:LBonOPT}
	\end{equation}
	In the remainder we compare $\EX[ALG]$ with $\EX[\OPT]$ component-wise.

	We can lower bound $\EX[\OPT_M]$ by $\sum_{v\in M} p_v$ using Equation~\eqref{eq_opt_partitioning_1}.
	By definition of~$M$, it holds that $\EX[\OPT_M] \ge \sum_{v \in M} p_v \ge d \cdot |M|$. Thus,
	\begin{equation}\label{gt_eq_2}
		|M| \le \frac{1}{d} \cdot \EX[\OPT_M].
	\end{equation}

	Next, we compare $|V_1| + \sum_{v \in V_0} p_v$ with $\EX[\OPT_{V_1 \cup V_0}]$.
	For this purpose, let $G[V_1 \cup V_0]$ be the subgraph of $G$ induced by $V_1 \cup V_0$,
	and let $G'[V_1 \cup V_0]$ be the bipartite graph that is created by removing all edges between elements of $V_1$ from $G[V_1\cup V_0]$.
	It follows from similar arguments as in~\cite[Theorem~2]{nemhauser1975} that $V_1$ is a minimum vertex cover of $G'[V_1\cup V_0]$. 
	(See
	Appendix~\ref{appThresholdBipartiteLemma}.)
	This allows us to apply the famous K\H{o}nig-Egerv\'{a}ry theorem~\cite{SchrijverBook}.
	By the latter there is a matching $h$ mapping each $v \in V_1$ to a distinct $h(v) \in V_0$ with $\{v,h(v)\} \in E$.
	Denoting $S = \{h(v) \mid v \in V_1 \}$, we can infer $\EX[\OPT_{V_1 \cup V_0}] \ge \EX[\OPT_{V_1 \cup S}] + \EX[\OPT_{V_0 \setminus S}]$.
			  
	Any feasible solution must query at least one endpoint of all edges of the form $\{v,h(v)\}$.
	This implies $\EX[\OPT_{V_1 \cup S}] \ge |V_1|$.
	Since additionally $p_{h(v)} \le d$ for each $h(v) \in S$, we get
	\begin{equation}\label{gt_eq_3}
		\sum_{v\in V_1} \left( 1 + p_{h(v)} \right) \le (1+d) \cdot |V_1| \le (1+d) \cdot \EX[OPT_{V_1\cup S}].
	\end{equation}
	By lower bounding $\EX[OPT_{V_0\setminus S}]$ with $\sum_{v\in V_0\setminus S} p_v$ and using \eqref{gt_eq_3}, we get
	\begin{align}\label{gt_eq_4}
			|V_1| + \sum_{v \in V_0} p_v & =  \sum_{v\in V_1} \left(1 + p_{h(v)}\right) + \sum_{v\in V_0\setminus S} p_v\\
			& \le (1+d) \cdot \EX[\OPT_{V_1\cup S}] + \EX[\OPT_{V_0 \setminus S}] 		\notag
			\le  (1+d) \cdot \EX[\OPT_{V_1\cup V_0}].								\notag
	\end{align}
	
	Finally, consider the term $|VC'| + \sum_{v \in V_{1/2} \setminus VC'} p_v$. 
	Let $VC^*$ be a minimum cardinality vertex cover for $G[V_{1/2}]$.
	Then, it holds $|VC^*| \ge \frac{1}{2} \cdot |V_{1/2}|$.
	This is, because in the optimal basic feasible solution to the LP relaxation $x^*$,  
	each vertex in $V_{1/2}$ has a value of $\frac{1}{2}$.
	A vertex cover with $|VC^*| < \frac{1}{2} \cdot |V_{1/2}| $ would contradict the optimality of $x^*$.
	{The following part of the analysis crucially relies on $|VC^*| \ge \frac{1}{2} \cdot |V_{1/2}| $, which is the reason why \threshold executes the LP relaxation-based preprocessing before applying the $\alpha$-approximation.}
			  
	Now, the expected cost of the algorithm for the subgraph $G[V_{1/2}]$ is $|VC'| + \sum_{v\in I'} p_v \le |VC'| + d\cdot |I'|$ with $I' = V_{1/2} \setminus VC'$.
	Since $|VC'| \ge |VC^*| \ge \frac{1}{2} \cdot |V_{1/2}| $, there is a tradeoff between the quality of $|VC'|$ and the additional cost of $d\cdot |I'|$.
	If $|VC'|$ is close to $\frac{1}{2} \cdot |V_{1/2}| $, then it is close to $|VC^*|$ but, on the other hand, $|I'|$ then is close to $\frac{1}{2} \cdot |V_{1/2}| $, which means that the additional cost $d\cdot |I'|$ is high. 
	Vice versa, if the cost for $|VC'|$ is high because it is larger than $\frac{1}{2} \cdot |V_{1/2}| $, then $|I'|$ is close to zero and the additional cost $d\cdot |I'|$ is low.
	We exploit this tradeoff and upper bound $\frac{|VC'| + d\cdot |I'|}{|VC^*|}$ in terms of the approximation factor $\alpha$ of the vertex cover approximation.
	Assume that the approximation factor $\alpha$ is tight, i.e., $|VC'| = \alpha \cdot |VC^*|$.
	Since $d \le 1$, this is the worst case for the ratio $\frac{|VC'| + d\cdot |I'|}{|VC^*|}$.
	(In other words, if the approximation factor was not tight, we
	could replace $\alpha$ by the approximation factor that is actually achieved and carry out the following
	calculations with that smaller value of $\alpha$ instead, yielding an even better bound.)
	Using $|VC'| = \alpha \cdot |VC^*|$ and $|VC^*| \ge \frac{1}{2} \cdot |V_{1/2}| $, we can derive
	\begin{align*}
		|I'| &= |V_{1/2}| - |VC'|
			 = |V_{1/2}| - \alpha \cdot |VC^*| 
			 \le (2-\alpha) \cdot |VC^*|.
	\end{align*}
	For the cost of the algorithm for subgraph $G[V_{1/2}]$ we get
	\begin{align*}
		|VC'| + d\cdot |I'| &\le \alpha \cdot |VC^*| + d \cdot (2-\alpha) \cdot |VC^*|
		= (\alpha + (2-\alpha) \cdot d) \cdot |VC^*|.
	\end{align*}
	Since $|VC^*| \le \EX[\OPT_{V_{1/2}}]$, we get
	\begin{equation}\label{gt_eq_5}
		|VC'| + d\cdot |I'| \le (\alpha + (2-\alpha) \cdot d) \cdot \EX[\OPT_{V_{1/2}}].
	\end{equation}

	Combining Equations~\eqref{eq_thres_alg},~\eqref{gt_eq_2},~\eqref{gt_eq_4} and~\eqref{gt_eq_5}, we can upper bound the cost of the algorithm:
	\begin{align*}
		\EX[ALG] 
		&= |M| + |V_1| + \sum_{v \in V_0} p_v + |VC'| + \sum_{v \in V_{1/2} \setminus VC'} p_v.\\
		&\le  \frac{1}{d} \cdot \EX[\OPT_M] + (1+d) \cdot \EX[\OPT_{V_1 \cup V_0}] + (\alpha + (2-\alpha) \cdot d)) \cdot \EX[\OPT_{V_{1/2}}]\\
		&\le  \max\left\{\frac{1}{d},(1+d),(\alpha + (2-\alpha) \cdot d)\right\} \cdot \EX[\OPT],
	\end{align*}
	where the last inequality follows from the lower bound on $\OPT$ in \eqref{eq:LBonOPT}. Observe that for any $d \in [0,1]$ and $\alpha \in [1,2]$, it holds that $(\alpha + (2-\alpha) \cdot d) \ge (1+d)$. {We conclude with 
	$
	\EX[ALG] \le \max\{\frac{1}{d},(\alpha + (2-\alpha) \cdot d)\} \cdot \EX[\OPT],
	$
	which implies the theorem.}
\end{proof}

In Appendix~\ref{appThreshold_tightness}, we show that the analysis of \threshold is tight.
To benefit from a better approximation factor than $\alpha=2$ for solving the minimum vertex cover problem, we would need to know in advance the specialized graph class on which we want to solve this subproblem. 
In some cases, we can benefit from optimal or approximation algorithms{, e.g., using \threshold with the PTAS for planar graphs~\cite{Bar1982} allows us to achieve a
competitive ratio of at most $1.618 + \epsilon$, for any $\epsilon > 0$, if the input graph is planar.}
\subsection{Orienting hypergraphs}
\label{sec:orienting-hyperedges}
 
\begin{theorem}
	\label{theorem_threshold_hyper}
	Given an $\alpha$-approximation with $1\le\alpha\le 2$ for the vertex cover problem
	(on the induced subgraph $\bar{G}[V_{1/2}]$ of the vertex cover instance given by Definition~\ref{def:vertex_cover_instance}, cf.~Line~\ref{line:threshold_solvevc}),
	a modified version of the \emph{\textsc{Threshold}} algorithm solves the hypergraph orientation problem with arbitrary query costs with competitive ratio
	$$
	R=	
	\frac{1}{2} \left(\alpha + \sqrt{\alpha ^2 + 4 (2-\alpha) (1 + \alpha \epsilon + (2- \alpha)\epsilon^2)} + (4 - 2 \alpha)\epsilon \right)
	$$
	with probability at least $1 - \delta$. Its running time is upper bounded by the complexity of the sampling procedure and the vertex cover black box procedure.
\end{theorem}
\begin{proof}

The modified algorithm works with the vertex cover instance $\bar G$ instead 
of the given hypergraph $H$, following
Definition~\ref{def:vertex_cover_instance}. In Line~\ref{line:threshold_addM},
we use $d(\alpha)=1/R +\epsilon$.
In Line~\ref{line:threshold_stage2}, we iteratively query mandatory vertices
until the instance is solved.
In addition, we use a random estimation
$Y_v$ of $p_v$, instead of the precise probability, using the procedure
described in Lemma~\ref{lem:sampling} with parameters $\epsilon$ and $\delta'$
such that $1- \delta = (1 - \delta')^n$. As a result, with probability at
least $1 - \delta$ we have that for every vertex $v$, the estimation $Y_v$ has
absolute error at most~$\epsilon$.  In case of this event we obtain the
following bound on the cost (which is optimized for the chosen value of $d$), namely $\EX[{\textrm{ALG}}] \leq \max \{\frac1{d- \epsilon}, (1+d  +\epsilon), (\alpha + (2 - \alpha)(d + \epsilon))\}  \cdot \EX[\OPT]$.
\end{proof}

Sorting a set of elements is equivalent to determining,
for each pair of elements, which of the two has smaller weight. Hence,
the problem of sorting multiple sets of elements with uncertain weights
is a special case of the graph orientation problem: For each set to be
sorted, the edge set of a complete graph on its elements is added to
a graph, and the resulting instance of the
graph orientation problem is then equivalent to the given instance
of the sorting problem.
In Appendix~\ref{app:constant}, we show the following theorem.
\begin{restatable}{theorem}{theoConstantSets}
	{
 	For the special cases of orienting $\mathcal{O}(\log \ |V|)$ hyperedges and sorting $\mathcal{O}(1)$ sets, \threshold can be applied with $\alpha = 1$ in polynomial running time.
	}
\end{restatable}

\section{Vertex cover-based algorithms: improved results for special cases}
\label{sec:VC-alg}

Consider an arbitrary vertex cover-based algorithm \ALG. 
It queries a vertex cover $VC$ in the first stage, and continues with elements of $V \setminus VC$ if they are mandatory. Thus,
	\begin{align*}
	\EX[\ALG] &= c(VC) + \sum_{v\in V\setminus VC} p_v \cdot c_v = \sum_{v\in VC} \left(p_v \cdot c_v + (1-p_v) \cdot c_v \right) + \sum_{v\in V\setminus VC} p_v \cdot c_v\\
	&= \sum_{v\in V} p_v \cdot c_v + \sum_{v \in VC} (1-p_v) \cdot c_v.
	\end{align*}
{Since the first term is independent of $VC$, \ALG\ is the best possible vertex cover-based algorithm 
if it minimizes $\sum_{v \in VC} (1-p_v) \cdot c_v$. We refer to this algorithm as \bestVC. 
}

To implement \bestVC, we need the exact value $p_v$, for all $v\in V$, and an optimal algorithm for computing a weighted vertex cover.
As mentioned in Section~\ref{sec:preliminaries},
the first problem is \#P-hard in hypergraphs (Appendix~\ref{app:sharpP}),
but it 
can be solved exactly in polynomial time for graphs.
{The weighted vertex cover problem can be solved optimally in polynomial time for bipartite graphs.}

In general, \bestVC  has competitive ratio at least $1.5$ (\Cref{theorem:VC_based_LB}). 
However, we show in the following that it is $4/3$-competitive for two special cases. 
{It remains open whether \bestVC still outperforms \threshold if the vertex cover is only approximated with a factor $\alpha > 1$.}


\subsection{A best possible algorithm for orienting bipartite graphs}

\begin{theorem}\label{Theorem:vertex_cover_bipartite}
	\emph{\bestVC} is $\frac{4}{3}$-competitive for the bipartite graph orientation problem.
\end{theorem}

\begin{proof}
In Appendix~\ref{subsec:star:subproblem}, we show that \bestVC is $\frac{4}{3}$-competitive for the problem of orienting stars if both vertex cover-based algorithms (either querying the leaves or the center first) have the same expected cost.
In this proof, we divide the instance into subproblems that fulfill these requirements, and use the result for stars to infer $\frac{4}{3}$-competitiveness for bipartite graphs.

Let $VC$ be a minimum-weight vertex cover (with weights $c_v \cdot (1-p_v)$) as computed by \bestVC in the first phase.
By the K\H{o}nig-Egerv\'{a}ry theorem (e.g.,~\cite{SchrijverBook}), there is
a function $\pi:E \rightarrow \mathbb{R}$ with $\sum_{\{u,v\} \in E} \pi(u,v) \le c_v \cdot (1-p_v)$ for each $v \in V$.
By duality theory, 
the constraint is tight for each $v \in VC$, and $\pi(u,v) = 0$ holds if both $u$ and $v$ are in $VC$. 
Thus, we can interpret $\pi$ as a function that distributes the 
weight of each $v \in VC$ to its neighbors outside of $VC$.

For each $v \in VC$ and $u \in V \setminus VC$, let $\lambda_{u,v} := \frac{\pi(u,v)}{(1-p_u) \cdot c_u}$ denote the fraction of the weight of $u$ that is used by $\pi$ to cover the weight of $v$.
Moreover, for $u \in V \setminus VC$, let $\tau_u := 1 - \sum_{\{u, v\} \in E} \lambda_{u,v}$ be the fraction of the weight of $u$ that is not used by $\pi$ to cover the weight of any $v \in VC$.
Then, we can write the expected cost of \bestVC as follows:
\begin{equation}
\EX[\bestVC] = \sum_{v\in VC} \bigg(c_v+\sum_{u\in V\setminus VC} p_u \cdot \lambda_{uv} \cdot c_u \bigg) + \sum_{u\in V\setminus VC} p_u \cdot \tau_u \cdot c_u .
\end{equation}
Using Observation~\ref{obs:vertex_splitting}, we compare $\EX[\bestVC]$ with the expected optimum $\EX[\OPT']$ for an instance $G'=(V',E',c')$ that is created by splitting vertices. 
We modify the mandatory distribution as 
in Section~\ref{sec:lower-bounds-on-OPT}, which implies $\EX[\OPT] = \EX[\OPT']$. 
We add the following copies:
\begin{enumerate}
	\item For each $v \in VC$, we add a copy $v'$ of $v$ to $V'$ with $c'_{v'} = c_v$.
	\item For all $u \in V\setminus VC$ and $v \in VC$ with $\lambda_{u,v} > 0$,
	we add a copy $u_v'$ of $u$ to $V'$ with $c'_{u_v'} = \lambda_{u,v} \cdot c_u$.
	\item For each $u \in V\setminus VC$ with $\tau_u > 0$, we add a copy $u'$ of $u$ to $V'$ with $c'_{u'} = \tau_u \cdot c_u$.
\end{enumerate}
Let $p'_{v}$ denote the probability of $v$ being mandatory for instance $G'$.
By definition of the vertex split operation, we have $p_v = p'_{u}$ for each $v \in V$ and each copy $u \in V'$ of $v$.
Further, for each $v \in VC$, define $H'_v = \{u'_v \mid u \in V \text{ with } \lambda_{u,v} > 0\}$. 
By definition of $H'_v$ and $G'$, all $H'_v$ are pairwise disjoint. 
Let $\mathcal{H}' = \bigcup_{v \in VC} \left(H'_v \cup \{v'\}\right)$; then we can express $\EX[\bestVC]$ as follows:
\begin{align}
\EX[\bestVC] &= \sum_{v\in VC} \bigg(c_v+\sum_{u\in V\setminus VC} p_u \cdot \lambda_{u,v} \cdot c_u \bigg) + \sum_{u\in V\setminus VC} p_u \cdot \tau_u \cdot c_u \nonumber \\
&=  \sum_{v\in VC} \bigg(c'_{v'}+\sum_{u\in H'_v} p_u' \cdot c'_u \bigg) + \sum_{u \in V' \setminus \mathcal{H'}} p'_u \cdot c'_u. \nonumber 
\end{align}
Similarly, we can lower bound $\EX[\OPT']$ using Lemma~\ref{lemma_opt_partitioning}:
\begin{equation}\label{eq_vc_opt_lb}
	\smash{
	\EX[\OPT'] \ge \hspace*{-0.5ex} \sum_{v\in VC} \EX[\OPT'_{H'_v\cup \{v'\}}] + \hspace*{-2ex}\sum_{u \in V' \setminus \mathcal{H'}} \hspace*{-1ex}\EX[\OPT'_{\{u\}}]
	\ge  \sum_{v\in VC} \EX[\OPT'_{H'_v\cup \{v'\}}] + \hspace*{-2ex} \sum_{u \in V' \setminus \mathcal{H'}} \hspace*{-1ex} p_u' \cdot c'_u.}\hspace*{1.5ex}
\end{equation}
Since the term $\sum_{u \in V' \setminus \mathcal{H'}} p_u' \cdot c'_u$ shows up in both inequalities, it remains, for each $v\in VC$, to bound $c'_{v'}+\sum_{u\in H'_v} p_u' \cdot c'_u$ in terms of  $\EX[\OPT'_{H'_v\cup \{v'\}}]$.
By definition of $H'_v$, we have $ (1-p'_{v'}) \cdot c'_{v'} = \sum_{u \in H_v'} (1-p_u') \cdot c'_u$, which implies
\begin{equation}\label{eq_secVC_1}
c'_{v'}+\sum_{u\in H'_v} p'_{u} \cdot c'_u = p'_{v'} \cdot c'_{v'} + \sum_{u\in H'_v} c'_u.
\end{equation}
The value $\EX[\OPT'_{H'_v \cup \{v'\}}]$ corresponds to the expected optimum for the subproblem which considers the subgraph induced by $H'_v \cup \{v'\}$ (which is a star), uses $p'_u$ as the mandatory probability for each $u \in H'_v \cup \{v'\}$, and uses $c'_u$ as the query cost of each $u \in H'_v \cup \{v'\}$.
For this subproblem, $c'_{v'}+\sum_{u\in H'_v} p'_{u} \cdot c'_u $ corresponds to the expected cost of the vertex cover-based algorithm that queries vertex cover $\{v'\}$ in the first stage.
Furthermore, $ p'_{v'} \cdot c'_{v'} + \sum_{u\in H'_v} c'_u$ corresponds to the expected cost of the vertex-cover based algorithm that queries vertex cover $H'_v$ in the first stage. 
In summary, we have a star orientation subproblem for which both vertex cover-based algorithms (querying the leaves or the center first) have the same expected cost (cf. Equation~\eqref{eq_secVC_1}).
As a major technical step, we show that \bestVC is $\frac{4}{3}$-competitive for such subproblems, which implies
$$c'_{v'}+\sum_{u\in H'_v} p'_u \cdot c'_u \le \frac{4}{3} \cdot \EX[\OPT'_{H'_v \cup \{v'\}}].$$  
A corresponding lemma is proven in \cref{subsec:star:subproblem}.
We remark that the lemma requires $p'_{v'}$ to be independent of each $p'_{u}$ with $u \in  H'_v$; otherwise the subproblem does not correspond to the star orientation problem. 
As the input graph is bipartite, such independence follows by definition.

Using this inequality and Equation~\eqref{eq_vc_opt_lb}, we conclude that \bestVC is $4/3$-competitive.
\begin{align*}
\EX[\bestVC] &= \sum_{v\in VC} \bigg(c'_{v'}+\sum_{u\in H'_v} p'_u \cdot c'_u \bigg) + \sum_{u \in V' \setminus \mathcal{H'}} p'_u \cdot c'_u \\
&\le \frac{4}{3} \cdot \sum_{v\in VC} \EX[\OPT'_{H'_v\cup\{v'\}}] + \sum_{u \in V' \setminus \mathcal{H'}} p'_u \cdot c'_u 
\le \frac{4}{3} \cdot \EX[\OPT'] = \frac{4}{3} \cdot \EX[\OPT] \qedhere
\end{align*}
\end{proof}


\subsection{An almost best possible algorithm for orienting a hyperedge }

\begin{restatable}[]{theorem}{theoremVCSingleSet}
\label{thm:VC-single-set}
\emph{\bestVC} has a competitive ratio at most $\min\{\frac{4}{3},\frac{n+1}{n}\}$ for the hypergraph orientation problem on a single hyperedge with $n\geq 2$ vertices and uniform query costs.
For a hyperedge with only two vertices, the algorithm is $1.207$-competitive.
\end{restatable}

The proof is in Appendix~\ref{subsec:proof-vc-single-set}. 
Our analysis improves upon a ${(n+1)}/{n}$-competitive algorithm by Chaplick et al.~\cite{chaplick20stochasticLATIN} in  case that the hyperedge has two or three vertices. Moreover, we show that this is near-optimal:
It is not hard to show a matching lower bound for two vertices and, due to Theorem~\ref{thm:generalLB43} and the theorem below, this is the best possible for three vertices, and in general the difference between the upper and lower bounds is less than 4\%.

\begin{restatable}[]{theorem}{theoremLBSingleSetUniform}
\label{thm:lb1setuniform}
Any algorithm for orienting a single hyperedge with $n+1 \geq 2$ vertices has competitive ratio at least ${n^2}/({n^2-n+1})$, even for uniform query costs.
\end{restatable}

Note that Theorem~\ref{thm:VC-single-set} is in contrast to the problem of orienting a hyperedge with arbitrary query costs:
In this setting,~\cite{chaplick20stochasticLATIN} showed that the algorithm is $1.5$-competitive, which matches the corresponding lower bound for vertex cover-based algorithms of Theorem~\ref{theorem:VC_based_LB}.

\section{Conclusion}
\label{sec:conclusion}

In this paper, we present algorithms for the (hyper)graph orientation problem under stochastic explorable uncertainty.
It remains open to determine the competitive ratio of \bestVC for the general (hyper)graph orientation problem, and to investigate how the algorithm behaves if it has to rely on an $\alpha$-approximation to solve the vertex cover subproblem.
In this context, one can consider the resulting algorithm as a standalone algorithm, or as a subroutine for \threshold.
Our analysis suggests that, to achieve a competitive ratio better than~$1.5$, algorithms have to employ more adaptivity; exploiting this possibility remains an open problem.  
Finally, it would be interesting to characterize the vertex cover instances arising in our \threshold algorithm. In addition to the relevance from a combinatorial point of view, such a characterization may allow an improved $\alpha$-approximation algorithm for those instances.

\bibliography{graph-explore}

\begin{thebibliography}{10}

\bibitem{AdamczykSW16}
Marek Adamczyk, Maxim Sviridenko, and Justin Ward.
\newblock Submodular stochastic probing on matroids.
\newblock {\em Math. Oper. Res.}, 41(3):1022--1038, 2016.

\bibitem{albersE2020}
Susanne Albers and Alexander Eckl.
\newblock Explorable uncertainty in scheduling with non-uniform testing times.
\newblock In {\em WAOA}, 2020.

\bibitem{arantes18schedulingqueries}
Luciana Arantes, Evripidis Bampis, Alexander~V. Kononov, Manthos Letsios,
  Giorgio Lucarelli, and Pierre Sens.
\newblock Scheduling under uncertainty: A query-based approach.
\newblock In {\em IJCAI 2018: 27th International Joint Conference on Artificial
  Intelligence}, pages 4646--4652, 2018.
\newblock \href {https://doi.org/10.24963/ijcai.2018/646}
  {\path{doi:10.24963/ijcai.2018/646}}.

\bibitem{Assadi2021}
Sepehr Assadi, Deeparnab Chakrabarty, and Sanjeev Khanna.
\newblock Graph connectivity and single element recovery via linear and {OR}
  queries.
\newblock In {\em ESA 2021: 29th Annual European Symposium on Algorithms},
  volume 204 of {\em LIPIcs}. Schloss Dagstuhl - Leibniz-Zentrum f{\"{u}}r
  Informatik, 2021.

\bibitem{AssadiKL19}
Sepehr Assadi, Sanjeev Khanna, and Yang Li.
\newblock The stochastic matching problem with (very) few queries.
\newblock {\em {ACM} Trans. Economics and Comput.}, 7(3):16:1--16:19, 2019.

\bibitem{BansalGLMNR12}
Nikhil Bansal, Anupam Gupta, Jian Li, Juli{\'{a}}n Mestre, Viswanath Nagarajan,
  and Atri Rudra.
\newblock When {LP} is the cure for your matching woes: Improved bounds for
  stochastic matchings.
\newblock {\em Algorithmica}, 63(4):733--762, 2012.

\bibitem{BansalN15}
Nikhil Bansal and Viswanath Nagarajan.
\newblock On the adaptivity gap of stochastic orienteering.
\newblock {\em Math. Program.}, 154(1-2):145--172, 2015.

\bibitem{baryehuda04}
Reuven Bar{-}Yehuda, Keren Bendel, Ari Freund, and Dror Rawitz.
\newblock Local ratio: {A} unified framework for approxmation algrithms. {In}
  memoriam: {Shimon} {Even} 1935-2004.
\newblock {\em {ACM} Comput. Surv.}, 36(4):422--463, 2004.
\newblock \href {https://doi.org/10.1145/1041680.1041683}
  {\path{doi:10.1145/1041680.1041683}}.

\bibitem{Bar1982}
Reuven Bar{-}Yehuda and Shimon Even.
\newblock On approximating a vertex cover for planar graphs.
\newblock In Harry~R. Lewis, Barbara~B. Simons, Walter~A. Burkhard, and
  Lawrence~H. Landweber, editors, {\em Proceedings of the 14th Annual {ACM}
  Symposium on Theory of Computing, May 5-7, 1982, San Francisco, California,
  {USA}}, pages 303--309. {ACM}, 1982.
\newblock \href {https://doi.org/10.1145/800070.802205}
  {\path{doi:10.1145/800070.802205}}.

\bibitem{Beame2018}
Paul Beame, Sariel Har{-}Peled, Sivaramakrishnan~Natarajan Ramamoorthy, Cyrus
  Rashtchian, and Makrand Sinha.
\newblock Edge estimation with independent set oracles.
\newblock In Anna~R. Karlin, editor, {\em 9th Innovations in Theoretical
  Computer Science Conference, {ITCS} 2018, January 11-14, 2018, Cambridge, MA,
  {USA}}, volume~94 of {\em LIPIcs}, pages 38:1--38:21. Schloss Dagstuhl -
  Leibniz-Zentrum f{\"{u}}r Informatik, 2018.
\newblock \href {https://doi.org/10.4230/LIPIcs.ITCS.2018.38}
  {\path{doi:10.4230/LIPIcs.ITCS.2018.38}}.

\bibitem{BehnezhadFHR19}
Soheil Behnezhad, Alireza Farhadi, MohammadTaghi Hajiaghayi, and Nima Reyhani.
\newblock Stochastic matching with few queries: New algorithms and tools.
\newblock In {\em {SODA}}, pages 2855--2874. {SIAM}, 2019.

\bibitem{BlumDHPSS20}
Avrim Blum, John~P. Dickerson, Nika Haghtalab, Ariel~D. Procaccia, Tuomas
  Sandholm, and Ankit Sharma.
\newblock Ignorance is almost bliss: Near-optimal stochastic matching with few
  queries.
\newblock {\em Oper. Res.}, 68(1):16--34, 2020.

\bibitem{bruce05uncertainty}
Richard Bruce, Michael Hoffmann, Danny Krizanc, and Rajeev Raman.
\newblock Efficient update strategies for geometric computing with uncertainty.
\newblock {\em Theory of Computing Systems}, 38(4):411--423, 2005.
\newblock \href {https://doi.org/10.1007/s00224-004-1180-4}
  {\path{doi:10.1007/s00224-004-1180-4}}.

\bibitem{BubeckC12}
S{\'{e}}bastien Bubeck and Nicol{\`{o}} Cesa{-}Bianchi.
\newblock Regret analysis of stochastic and nonstochastic multi-armed bandit
  problems.
\newblock {\em Foundations and Trends in Machine Learning}, 5(1):1--122, 2012.

\bibitem{chaplick20stochasticLATIN}
Steven Chaplick, Magnu\'{u}s~M. Halld\'{o}rsson, Murilo~S. de~Lima, and Tigran
  Tonoyan.
\newblock Query minimization under stochastic uncertainty.
\newblock In Y.~Kohayakawa and F.~K. Miyazawa, editors, {\em LATIN 2020: 14th
  Latin American Theoretical Informatics Symposium}, volume 12118 of {\em
  Lecture Notes in Computer Science}, pages 181--193. Springer Berlin
  Heidelberg, 2020.

\bibitem{ChenIKMR09}
Ning Chen, Nicole Immorlica, Anna~R. Karlin, Mohammad Mahdian, and Atri Rudra.
\newblock Approximating matches made in heaven.
\newblock In {\em {ICALP} {(1)}}, volume 5555 of {\em Lecture Notes in Computer
  Science}, pages 266--278. Springer, 2009.

\bibitem{Chen2020}
Xi~Chen, Amit Levi, and Erik Waingarten.
\newblock Nearly optimal edge estimation with independent set queries.
\newblock In Shuchi Chawla, editor, {\em Proceedings of the 2020 {ACM-SIAM}
  Symposium on Discrete Algorithms, {SODA} 2020, Salt Lake City, UT, USA,
  January 5-8, 2020}, pages 2916--2935. {SIAM}, 2020.
\newblock \href {https://doi.org/10.1137/1.9781611975994.177}
  {\path{doi:10.1137/1.9781611975994.177}}.

\bibitem{chlebik2007}
Miroslav Chlebik and Janka Chleb{\'\i}kov{\'a}.
\newblock The complexity of combinatorial optimization problems on
  d-dimensional boxes.
\newblock {\em SIAM Journal on Discrete Mathematics}, 21(1):158--169, 2007.

\bibitem{DeanGV08}
Brian~C. Dean, Michel~X. Goemans, and Jan Vondr{\'{a}}k.
\newblock Approximating the stochastic knapsack problem: The benefit of
  adaptivity.
\newblock {\em Math. Oper. Res.}, 33(4):945--964, 2008.

\bibitem{DurrEMM20}
Christoph D{\"{u}}rr, Thomas Erlebach, Nicole Megow, and Julie Mei{\ss}ner.
\newblock An adversarial model for scheduling with testing.
\newblock {\em Algorithmica}, 82(12):3630--3675, 2020.

\bibitem{erlebach14mstverification}
Thomas Erlebach and Michael Hoffmann.
\newblock Minimum spanning tree verification under uncertainty.
\newblock In D.~Kratsch and I.~Todinca, editors, {\em WG 2014: International
  Workshop on Graph-Theoretic Concepts in Computer Science}, volume 8747 of
  {\em Lecture Notes in Computer Science}, pages 164--175. Springer Berlin
  Heidelberg, 2014.
\newblock \href {https://doi.org/10.1007/978-3-319-12340-0_14}
  {\path{doi:10.1007/978-3-319-12340-0_14}}.

\bibitem{ErlebachHL2021}
Thomas Erlebach, Michael Hoffmann, and Murilo~S. de~Lima.
\newblock Round-competitive algorithms for uncertainty problems with parallel
  queries.
\newblock In Markus Bl{\"{a}}ser and Benjamin Monmege, editors, {\em {STACS}
  2021: 38th International Symposium on Theoretical Aspects of Computer
  Science}, volume 187 of {\em LIPIcs}, pages 27:1--27:18. Schloss Dagstuhl -
  Leibniz-Zentrum f{\"{u}}r Informatik, 2021.
\newblock \href {https://doi.org/10.4230/LIPIcs.STACS.2021.27}
  {\path{doi:10.4230/LIPIcs.STACS.2021.27}}.

\bibitem{ErlebachHLMS-arxiv2020}
Thomas Erlebach, Michael Hoffmann, Murilo~S. de~Lima, Nicole Megow, and Jens
  Schl{\"{o}}ter.
\newblock Untrusted predictions improve trustable query policies.
\newblock {\em CoRR}, abs/2011.07385, 2020.

\bibitem{erlebach08steiner_uncertainty}
Thomas Erlebach, Michael Hoffmann, Danny Krizanc, Mat{\'{u}}s Mihal{\'{a}}k,
  and Rajeev Raman.
\newblock Computing minimum spanning trees with uncertainty.
\newblock In {\em STACS'08: 25th International Symposium on Theoretical Aspects
  of Computer Science}, volume~1 of {\em LIPIcs}, pages 277--288. Schloss
  Dagstuhl - Leibniz-Zentrum fuer Informatik, Germany, 2008.
\newblock \href {https://doi.org/10.4230/LIPIcs.STACS.2008.1358}
  {\path{doi:10.4230/LIPIcs.STACS.2008.1358}}.

\bibitem{feder07pathsqueires}
Tom{\'{a}}s Feder, Rajeev Motwani, Liadan O'Callaghan, Chris Olston, and Rina
  Panigrahy.
\newblock Computing shortest paths with uncertainty.
\newblock {\em Journal of Algorithms}, 62(1):1--18, 2007.
\newblock \href {https://doi.org/10.1016/j.jalgor.2004.07.005}
  {\path{doi:10.1016/j.jalgor.2004.07.005}}.

\bibitem{feder03medianqueries}
Tom{\'{a}}s Feder, Rajeev Motwani, Rina Panigrahy, Chris Olston, and Jennifer
  Widom.
\newblock Computing the median with uncertainty.
\newblock {\em SIAM Journal on Computing}, 32(2):538--547, 2003.
\newblock \href {https://doi.org/10.1137/S0097539701395668}
  {\path{doi:10.1137/S0097539701395668}}.

\bibitem{focke20mstexp}
Jacob Focke, Nicole Megow, and Julie Mei{\ss}ner.
\newblock Minimum spanning tree under explorable uncertainty in theory and
  experiments.
\newblock {\em {ACM} J. Exp. Algorithmics}, 25:1--20, 2020.
\newblock \href {https://doi.org/10.1145/3422371} {\path{doi:10.1145/3422371}}.

\bibitem{FrankBook}
Andr\'as Frank.
\newblock {\em Connections in Combinatorial Optimization}.
\newblock Oxford University Press, USA, 2011.

\bibitem{GittinsGW11-book}
John Gittins, Kevin Glazebrook, and Richard Weber.
\newblock {\em Multi-armed Bandit Allocation Indices}.
\newblock Wiley, 2nd edition, 2011.

\bibitem{goerigk15knapsackqueries}
Marc Goerigk, Manoj Gupta, Jonas Ide, Anita Sch{\"{o}}bel, and Sandeep Sen.
\newblock The robust knapsack problem with queries.
\newblock {\em Computers \& Operations Research}, 55:12--22, 2015.
\newblock \href {https://doi.org/10.1016/j.cor.2014.09.010}
  {\path{doi:10.1016/j.cor.2014.09.010}}.

\bibitem{Goldreich2017}
Oded Goldreich.
\newblock {\em Introduction to Property Testing}.
\newblock Cambridge University Press, 2017.
\newblock \href {https://doi.org/10.1017/9781108135252}
  {\path{doi:10.1017/9781108135252}}.

\bibitem{gupta2019markovian}
Anupam Gupta, Haotian Jiang, Ziv Scully, and Sahil Singla.
\newblock The markovian price of information.
\newblock In {\em International Conference on Integer Programming and
  Combinatorial Optimization}, pages 233--246. Springer, 2019.

\bibitem{GuptaKNR15}
Anupam Gupta, Ravishankar Krishnaswamy, Viswanath Nagarajan, and R.~Ravi.
\newblock Running errands in time: Approximation algorithms for stochastic
  orienteering.
\newblock {\em Math. Oper. Res.}, 40(1):56--79, 2015.

\bibitem{GuptaN13}
Anupam Gupta and Viswanath Nagarajan.
\newblock A stochastic probing problem with applications.
\newblock In {\em {IPCO}}, volume 7801 of {\em Lecture Notes in Computer
  Science}, pages 205--216. Springer, 2013.

\bibitem{GuptaNS16}
Anupam Gupta, Viswanath Nagarajan, and Sahil Singla.
\newblock Algorithms and adaptivity gaps for stochastic probing.
\newblock In {\em {SODA}}, pages 1731--1747. {SIAM}, 2016.

\bibitem{gupta16queryselection}
Manoj Gupta, Yogish Sabharwal, and Sandeep Sen.
\newblock The update complexity of selection and related problems.
\newblock {\em Theory of Computing Systems}, 59(1):112--132, 2016.
\newblock \href {https://doi.org/10.1007/s00224-015-9664-y}
  {\path{doi:10.1007/s00224-015-9664-y}}.

\bibitem{halldorssonL21sortingfull}
Magn{\'{u}}s~M. Halld{\'{o}}rsson and Murilo~Santos de~Lima.
\newblock Query-competitive sorting with uncertainty.
\newblock {\em Theor. Comput. Sci.}, 867:50--67, 2021.
\newblock \href {https://doi.org/10.1016/j.tcs.2021.03.021}
  {\path{doi:10.1016/j.tcs.2021.03.021}}.

\bibitem{halperin2002improved}
Eran Halperin.
\newblock Improved approximation algorithms for the vertex cover problem in
  graphs and hypergraphs.
\newblock {\em SIAM Journal on Computing}, 31(5):1608--1623, 2002.

\bibitem{kahan91}
Simon Kahan.
\newblock A model for data in motion.
\newblock In {\em STOC'91: 23rd Annual ACM Symposium on Theory of Computing},
  pages 265--277, 1991.
\newblock \href {https://doi.org/10.1145/103418.103449}
  {\path{doi:10.1145/103418.103449}}.

\bibitem{khanna01queries}
Sanjeev Khanna and Wang-Chiew Tan.
\newblock On computing functions with uncertainty.
\newblock In {\em PODS'01: 20th ACM SIGMOD-SIGACT-SIGART Symposium on
  Principles of Database Systems}, pages 171--182, 2001.
\newblock \href {https://doi.org/10.1145/375551.375577}
  {\path{doi:10.1145/375551.375577}}.

\bibitem{LeviMS19}
Retsef Levi, Thomas~L. Magnanti, and Yaron Shaposhnik.
\newblock Scheduling with testing.
\newblock {\em Manag. Sci.}, 65(2):776--793, 2019.

\bibitem{Ma18}
Will Ma.
\newblock Improvements and generalizations of stochastic knapsack and markovian
  bandits approximation algorithms.
\newblock {\em Math. Oper. Res.}, 43(3):789--812, 2018.

\bibitem{Mazzawi2010}
Hanna Mazzawi.
\newblock Optimally reconstructing weighted graphs using queries.
\newblock In Moses Charikar, editor, {\em Proceedings of the Twenty-First
  Annual {ACM-SIAM} Symposium on Discrete Algorithms, {SODA} 2010, Austin,
  Texas, USA, January 17-19, 2010}, pages 608--615. {SIAM}, 2010.
\newblock \href {https://doi.org/10.1137/1.9781611973075.51}
  {\path{doi:10.1137/1.9781611973075.51}}.

\bibitem{megow17mst}
Nicole Megow, Julie Mei{\ss}ner, and Martin Skutella.
\newblock Randomization helps computing a minimum spanning tree under
  uncertainty.
\newblock {\em SIAM Journal on Computing}, 46(4):1217--1240, 2017.
\newblock \href {https://doi.org/10.1137/16M1088375}
  {\path{doi:10.1137/16M1088375}}.

\bibitem{MerinoS19}
Arturo~I. Merino and Jos{\'{e}}~A. Soto.
\newblock The minimum cost query problem on matroids with uncertainty areas.
\newblock In {\em Proceedings of {ICALP}}, volume 132 of {\em LIPIcs}, pages
  83:1--83:14. Schloss Dagstuhl - Leibniz-Zentrum f{\"{u}}r Informatik, 2019.

\bibitem{nemhauser1975}
George~L. Nemhauser and Leslie E.~Trotter Jr.
\newblock Vertex packings: Structural properties and algorithms.
\newblock {\em Mathematical Programming}, 8:232--248, 1975.

\bibitem{Nisan2021}
Noam Nisan.
\newblock The demand query model for bipartite matching.
\newblock In {\em SODA}, 2021.

\bibitem{olston2000queries}
Chris Olston and Jennifer Widom.
\newblock Offering a precision-performance tradeoff for aggregation queries
  over replicated data.
\newblock In {\em VLDB 2000: 26th International Conference on Very Large Data
  Bases}, pages 144--155, 2000.
\newblock URL: \url{http://ilpubs.stanford.edu:8090/437/}.

\bibitem{robbins1939}
Herbert~E.\ Robbins.
\newblock A theorem on graphs, with an application to a problem of traffic
  control.
\newblock {\em The American Mathematical Monthly}, 46(5):281--283, 1939.
\newblock URL: \url{http://www.jstor.org/stable/2303897}.

\bibitem{Rubinstein2018}
Aviad Rubinstein, Tselil Schramm, and S.~Matthew Weinberg.
\newblock Computing exact minimum cuts without knowing the graph.
\newblock In Anna~R. Karlin, editor, {\em 9th Innovations in Theoretical
  Computer Science Conference, {ITCS} 2018, January 11-14, 2018, Cambridge, MA,
  {USA}}, volume~94 of {\em LIPIcs}, pages 39:1--39:16. Schloss Dagstuhl -
  Leibniz-Zentrum f{\"{u}}r Informatik, 2018.
\newblock \href {https://doi.org/10.4230/LIPIcs.ITCS.2018.39}
  {\path{doi:10.4230/LIPIcs.ITCS.2018.39}}.

\bibitem{SchrijverBook}
Alexander Schrijver.
\newblock {\em Combinatorial Optimization $-$ Polyhedra and Efficiency}.
\newblock Springer, 2003.

\bibitem{singla2018price}
Sahil Singla.
\newblock The price of information in combinatorial optimization.
\newblock In {\em Proceedings of the Twenty-Ninth Annual ACM-SIAM Symposium on
  Discrete Algorithms}, pages 2523--2532. SIAM, 2018.

\bibitem{Thompson33}
William~R. Thompson.
\newblock On the likelihood that one unknown probability exceeds another in
  view of the evidence of two samples.
\newblock {\em Biometrika}, 25(3/4):285--294, 1933.

\bibitem{Weitzman1979}
Martin Weitzman.
\newblock Optimal search for the best alternative.
\newblock {\em Econometrica}, 47(3):641--54, 1979.

\bibitem{yannakakis_edge_1980}
Mihalis Yannakakis and Fanica Gavril.
\newblock Edge {Dominating} {Sets} in {Graphs}.
\newblock {\em SIAM Journal on Applied Mathematics}, 38(3):364--372, June 1980.
\newblock URL: \url{http://epubs.siam.org/doi/10.1137/0138030}, \href
  {https://doi.org/10.1137/0138030} {\path{doi:10.1137/0138030}}.

\end{thebibliography}

\appendix


\section{Appendix for Section~\ref{sec:preliminaries}}
\label{app:preliminaries}

\subsection{Proof of Lemma~\ref{lemma:mandatory}}
\lemmamandatory*

\begin{proof}
	If~$v$ is a minimum-weight vertex of~$F$ and $w_u \in I_v$ for another vertex $u\in F \setminus \{v\}$, then~$F$ cannot be oriented even if we query all vertices in $F \setminus \{v\}$.
	If~$v$ is not a minimum-weight vertex of~$F$ and contains the weight $w_u$ of the minimum-weight vertex $u$ of $F$, then~$F$ cannot be oriented even if we query all vertices in $F \setminus \{v\}$, as we cannot prove that $w_u \leq w_v$.
	
	If~$v$ is a minimum-weight vertex of~$F$, but $w_u \notin I_v$ for every $u \in F \setminus \{v\}$, then $F \setminus \{v\}$ is a feasible query set for orienting $F$.
	If~$v$ is not a minimum-weight vertex of~$F$ and does not contain the weight $w_u$ of a minimum-weight vertex of~$F$, then again $F \setminus \{v\}$ is a feasible query set for orienting $F$. If every hyperedge $F$ that contains $v$ falls into one of these two cases, then querying all vertices except $v$ is a feasible query set for the whole~instance.
\end{proof}

\subsection{Proof of \#P-hardness}
\label{app:sharpP}

\begin{theorem}
\label{thm:sharpPhard}
Computing $p_v$ for a hypergraph $H=(V,E)$ is \#P-hard, even if all hyperedges have size~$3$.
\end{theorem}

\begin{proof}
The proof consists of a reduction of the \#P-hard problem of counting the number of vertex covers of a given graph $G=(V,E)$.  We construct the following instance to the hypergraph orientation problem.  

Let $v'$ be a new vertex, which does not belong to $V$. We construct the hypergraph $H=(V\cup\{v'\}, E')$, where every edge $\{u,v\}\in E$ from the original graph corresponds to an hyperedge in $E'$ of the form $\{u,v,v'\}$.  The value $w_v$ associated to every vertex $v\in V$ follows a distribution on the interval $I_v = (i\varepsilon, 2+i\varepsilon)$ for an
arbitrary $\varepsilon$ with $0 < \varepsilon < 1/n$, where $i$ is the rank of $v$ for some arbitrary fixed ordering of $V$.  The role of the shift by $\varepsilon$ is to avoid inclusion of the intervals. The distribution on
$I_v$ is constructed such that the variable $v$ is less than~$1$ with probability $1/2$
and at least~$1$ with probability $1/2$. 
The value associated to vertex $v'$ belongs to the interval $[1,2+(n+1) \varepsilon]$.  

What is the probability	that $v'$ is mandatory? For every realization of the
vertices, it is only crucial which of them have value less than $1$.  Let
$T\subseteq V$ be the set of the vertices with a value below the threshold $1$.  Every
potential set $T$ has the same probability $1/2^n$. 

Now $v'$ becomes mandatory by a set $\{u,v,v'\}$ if and only if both
values $w_u,w_v$ are at least~$1$.  This means that $v'$ is \emph{not} mandatory if and
only if the set $T$ defined by the realization is a vertex cover
for the graph $G$.  In other words the probability that $v'$ is mandatory
equals $1-\ell/2^n$, where $\ell$ is the number of vertex covers of $G$. 
As a consequence, computing this probability is \#P-hard.
\end{proof}

\subsection{Proof of Lemma~\ref{lem:sampling}}
\lemsampling*

\begin{proof}
The algorithm consists of drawing independently $k =\lceil \ln(2/\delta)/(2\epsilon^2) \rceil$ realizations of the vertex values, and determining for each realization whether $v$ is mandatory by using Lemma~\ref{lemma:mandatory}.  
The output $y$ is the fraction of realizations where this event happened.

For the analysis, let $X_v^i$ be the random variable indicating whether $v$ is mandatory in the $i$-th realization. These variables are independent and have mean $p_v$.  Using Hoeffding's concentration bound we have
\begin{align*}
	{\mathbb P}\left[\left|\frac1k\sum_{i=1}^k X_v^i - p_v\right| > \epsilon\right] &< 2 \exp(-2k\epsilon^2)
\end{align*}
The choice of $k$ is motivated by the following equivalent inequalties.
\begin{align*}
	2\exp(-2k \epsilon^2) & \leq \delta
	\\
	-2k \epsilon^2 & \leq \ln(\delta/2)
	\\
	k &\geq \ln(2/\delta)/(2 \epsilon^2).
\end{align*}
For the time complexity we observe that verifying if a vertex is mandatory can be done in time $O(|V|^2)$.
Its practical implementation makes use of the given probability matrix. Let $t_1,\ldots,t_{2|V|}$ be the sorted elements of $\{\ell_v,r_v|v\in V\}$, defining elementary intervals of the form $(t_i, t_{i+1})$.  The given probability matrix specifies for every given vertex $v\in V$ the probability that its weight $w_v$ belongs to a given interval $(t_i,t_{i+1})$.  Instead of sampling the actual weights, we only sample the elementary intervals to which they belong. Hence, for a fixed sample, 
we know for each hyperedge $E$ in which elementary interval $M$ its minimum weight lies.
In this setting, the vertices that are mandatory because of $E$ can be determined as follows:
\begin{itemize}
	\item If at least two weights of vertices in $E$ belong to~$M$, then all vertices $v\in E$ with $M\subseteq I_v=(\ell_v,r_v)$ are mandatory.
	\item Otherwise, let $m\in E$ be the unique vertex whose weight lies in $M$. We then have:
	\begin{itemize}
		\item Every other vertex $u\in E$ with $M\subseteq I_u$ is mandatory.
		\item If some vertex $u\in E$ has its weight in $I_m$, then $m$ is also mandatory.
	\end{itemize}
\end{itemize}
The union of the mandatory vertices of all hyperedges then gives the set of all mandatory vertices.
\end{proof}


\subsection{Proof of Theorem~\ref{theorem:VC_based_LB}}

\theoremVCbasedLB*

\begin{proof}
	\begin{figure}[tbh]
		\centering
		\begin{tikzpicture}[line width = 0.3mm]
		\interval{$x$}{0}{3}{0}
		\interval{$y$}{1}{4}{-0.5}
		\interval{$z$}{2}{5}{-1}
		
		\draw[dotted] (1, 0.5) -- (1, -1.5);
		\draw[dotted] (2, 0.5) -- (2, -1.5);
		\draw[dotted] (3, 0.5) -- (3, -1.5);
		
		\node at (0.5, 0.2) {$\epsilon$};
		\node at (1.5, 0.25) {$0$};
		\node at (2.5, 0.23) {$1 - \epsilon$};
		
		\node at (1.5, -0.3) {$1/2$};
		\node at (2.5, -0.26) {$0$};
		\node at (3.5, -0.3) {$1/2$};
		
		\node at (2.5, -0.8) {$\epsilon$};
		\node at (4, -0.8) {$1 - \epsilon$};
		
		\node at (6.5, 0) {$c_x = k$};
		\node at (6.5, -0.5) {$c_y = 1$};
		\node at (6.5, -1) {$c_z = k$};
		\end{tikzpicture}
		\caption{Lower bound example for orienting a single hyperedge.}
		\label{Ex_1}
	\end{figure}
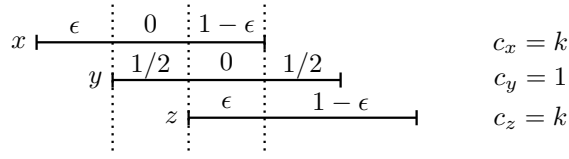
	
	\textbf{Case 1.} Consider the instance of Figure~\ref{Ex_1}.
	Each two-stage algorithm has to query a vertex cover for the vertex cover instance in the first stage. 
	Otherwise, it is not guaranteed that the problem can be solved in the second stage while only querying mandatory elements.
	Thus, each algorithm either queries $A_1 = \{x\}$, $A_2 = \{y, z\}$ or a superset of one of them in the first stage.
	Since querying supersets of $A_1$ or $A_2$ in the first stage never improves in comparison to just querying $A_1$ or $A_2$, we discard that option. 
	For $\epsilon$ tending to zero, the expected cost of the two-stage algorithms are
	$\EX[A_1] = \EX[A_2] = k + 1 + \frac{1}{2} \cdot k = \frac{3}{2} \cdot k + 1$,
	while the expected optimum is 
	$\EX[\OPT] = 1 + \frac{1}{2} \cdot k + \frac{1}{2} \cdot k = k+1$
	and therefore, for $k$ tending to infinity, the competitive ratio approaches $\frac{3}{2}$.
	\begin{figure}[tb]
		\centering
		\begin{tikzpicture}[line width = 0.3mm]
		\interval{$x_1$}{0}{3}{1}
		\path (-0.3, 0) -- (-0.3, 1) node[font=\normalsize, midway, sloped]{$\dots$};
		\interval{$x_k$}{0}{3}{0}
		\interval{$y$}{1}{4}{-0.5}
		\interval{$z_1$}{2}{5}{-1}
		\path (1.7, -1) -- (1.7, -2) node[font=\normalsize, midway, sloped]{$\dots$};
		\interval{$z_k$}{2}{5}{-2}
		
		\draw[dotted] (1, 1.5) -- (1, -2.5);
		\draw[dotted] (2, 1.5) -- (2, -2.5);
		\draw[dotted] (3, 1.5) -- (3, -2.5);
		
		\node at (0.5, 1.2) {$\epsilon$};
		\node at (1.5, 1.25) {$0$};
		\node at (2.5, 1.23) {$1 - \epsilon$};
		
		\node at (0.5, 0.2) {$\epsilon$};
		\node at (1.5, 0.25) {$0$};
		\node at (2.5, 0.23) {$1 - \epsilon$};
		
		\node at (1.5, -0.3) {$1/2$};
		\node at (2.5, -0.26) {$0$};
		\node at (3.5, -0.3) {$1/2$};
		
		\node at (2.5, -0.8) {$\epsilon$};
		\node at (4, -0.8) {$1 - \epsilon$};
		
		\node at (2.5, -1.8) {$\epsilon$};
		\node at (4, -1.8) {$1 - \epsilon$};
		\end{tikzpicture}
		\qquad
		\begin{tikzpicture}[thick, scale=0.5]
		\draw (0, 0) node[vertex, label=west:$x_1$]{} -- (4, 0) node[vertex, label=east:$z_1$]{};
		\draw (0, 0) -- (4, 1.5) node[vertex, label=east:$y$]{};
		\draw (0, -3) node[vertex, label=west:$x_k$]{} -- (4, -3) node[vertex, label=east:$z_k$]{};
		\draw (0, 0) -- (4, -3);
		\draw (0, -3) -- (4, 1.5);
		\draw (0, -3) -- (4, 0);
		
		\path (-0.8, 0) -- (-0.8, -3) node[font=\normalsize, midway, sloped]{$\dots$};
		\path (4.7, 0) -- (4.7, -3) node[font=\normalsize, midway, sloped]{$\dots$};
		
		\path (0.6, 0.5) -- (0.6, -3) node[font=\normalsize, midway, sloped]{$\dots$};
		\path (3.2, 0) -- (3.2, -3) node[font=\normalsize, midway, sloped]{$\dots$};
		\end{tikzpicture}
		\caption{Lower bound example for the hyperedge orientation problem.
			The vertex cover instance is shown in the complete bipartite graph at the right.}
		\label{Ex_2}
	\end{figure}
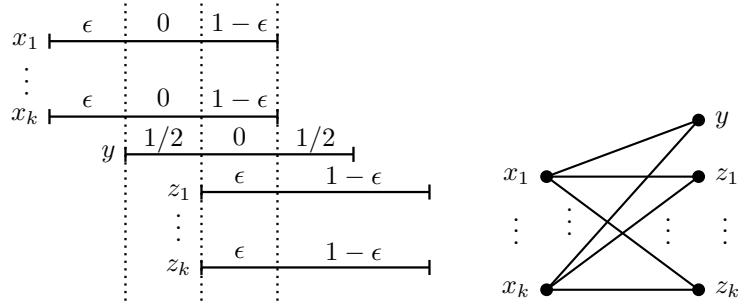

	\textbf{Case 2.} The argumentation is analogous to the proof of the first case but uses the instance in Figure~\ref{Ex_2}.
	The hyperedges are $S_1, \ldots, S_k$, with $S_i = \{x_i, y, z_1, \ldots, z_k\}$.
	
\textbf{Case~3.} The argumentation is analogous to the proof of Case~$1$ but uses the instance in Figure~\ref{Ex_3}.
		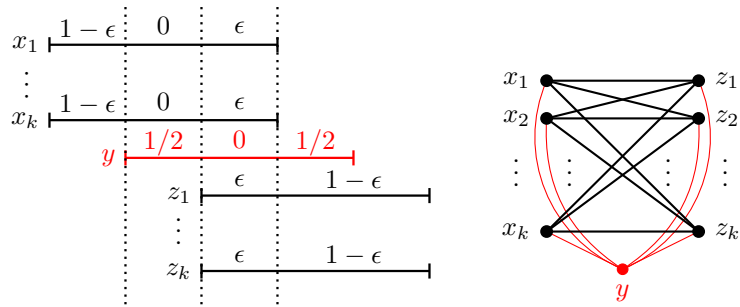
\begin{figure}[tb]
		\centering
		\begin{tikzpicture}[line width = 0.3mm]
		\interval{$x_1$}{0}{3}{1}
		\path (-0.3, 0) -- (-0.3, 1) node[font=\normalsize, midway, sloped]{$\dots$};
		\interval{$x_k$}{0}{3}{0}
		\intervalcolor{$y$}{1}{4}{-0.5}{red}
		\interval{$z_1$}{2}{5}{-1}
		\path (1.7, -1) -- (1.7, -2) node[font=\normalsize, midway, sloped]{$\dots$};
		\interval{$z_k$}{2}{5}{-2}
		
		\draw[dotted] (1, 1.5) -- (1, -2.5);
		\draw[dotted] (2, 1.5) -- (2, -2.5);
		\draw[dotted] (3, 1.5) -- (3, -2.5);
		
		\node at (0.5, 1.2) {$1- \epsilon$};
		\node at (1.5, 1.25) {$0$};
		\node at (2.5, 1.23) {$\epsilon$};
		
		\node at (0.5, 0.2) {$1 - \epsilon$};
		\node at (1.5, 0.25) {$0$};
		\node at (2.5, 0.23) {$\epsilon$};
		
		\node[red] at (1.5, -0.3) {$1/2$};
		\node[red] at (2.5, -0.26) {$0$};
		\node[red] at (3.5, -0.3) {$1/2$};
		
		\node at (2.5, -0.8) {$\epsilon$};
		\node at (4, -0.8) {$1 - \epsilon$};
		
		\node at (2.5, -1.8) {$\epsilon$};
		\node at (4, -1.8) {$1 - \epsilon$};
		\end{tikzpicture}
		\qquad
		\begin{tikzpicture}[thick, scale=0.5]
		\draw (0, 1) node[vertex, label=west:$x_1$]{} -- (4, 1) node[vertex, label=east:$z_1$]{};
		\draw (0, 0) node[vertex, label=west:$x_2$]{} -- (4, 0) node[vertex, label=east:$z_2$]{};
		\draw (0, -3) node[vertex, label=west:$x_k$]{} -- (4, -3) node[vertex, label=east:$z_k$]{};
		\draw (0, 0) -- (4, 1);
		\draw (0, 0) -- (4, -3);
		\draw (0, 1) -- (4, 0);
		\draw (0, 1) -- (4, -3);
		\draw (0, -3) -- (4, 1);
		\draw (0, -3) -- (4, 0);
		
		\begin{scope}[on background layer]
		\draw[red] (2, -4) node[red, fill=red, vertex, label=south:$y$]{} -- (0, -3);
		\draw[red] (2, -4) -- (4, -3);
		\draw[red] plot (2, -4) to [bend left=30] (0, 0);
		\draw[red] plot (2, -4) to [bend left=45] (0, 1);
		\draw[red] plot (2, -4) to [bend right=30] (4, 0);
		\draw[red] plot (2, -4) to [bend right=45] (4, 1);
		\end{scope}
		
		\path (-0.8, 0) -- (-0.8, -3) node[font=\normalsize, midway, sloped]{$\dots$};
		\path (4.7, 0) -- (4.7, -3) node[font=\normalsize, midway, sloped]{$\dots$};
		
		\path (0.6, 0) -- (0.6, -3) node[font=\normalsize, midway, sloped]{$\dots$};
		\path (3.2, 0) -- (3.2, -3) node[font=\normalsize, midway, sloped]{$\dots$};
		\end{tikzpicture}
		\caption{Lower bound example for the graph orientation problem with uniform query costs.
			The subgraph induced by $\{x_1, \ldots x_k, z_1, \ldots, z_k\}$ is complete bipartite, and $y$ is a universal vertex (adjacent to all others).}
		\label{Ex_3}
	\end{figure}
\end{proof}


\subsection{Lower bound for strict two-stage algorithms for minimum}
\label{app:strict2stage}

An algorithm $ALG$ is a \emph{strict two-stage algorithm} if in a first stage it non-adaptively queries a set of elements $U_1$ and in a second stage it non-adaptively queries a set of elements $U_2$.
The choice of $U_2$ can depend on the outcome of the queries for~$U_1$.
The algorithm must guarantee that after the second stage it can determine
the minimum weight element.

\begin{theorem}
No strict two-stage algorithm can have competitive ratio
$o(\log n)$ for the hypergraph orientation problem, even
for a single hyperedge with uniform query costs.
\end{theorem}

\begin{proof}
Consider the following instance:
\begin{itemize}
\item $I_1=(1,n+1)$
\item $I_i=(i,n+2)$ for $2\le i\le n$
\end{itemize}
The probability distributions are such that, for each interval $(a,b)$,
the probability is $1/2$ for the weight to be in $(a,a+1)$, and
$1/2$ for the weight to be in $(b-1,b)$. So the weight of each interval
is either at its left end or at its right end, and each of these events happens with probability~$1/2$.

If we query the intervals in order of left endpoint, we
are done as soon as one interval has its weight at its left end.
So each interval has probability $1/2$ to be the last one
we need to query, and the expected number of queries is~$2$.
So $\EX[\OPT]\le 2$.

Now consider an arbitrary strict two-stage algorithm $ALG$.
Assume that the algorithm queries $k$ intervals in the
first stage.

Consider the event $E$ that each of the intervals
queried by $ALG$ in the first stage has its weight at the
right end. The event $E$ occurs with probability $1/2^k$
and, if it occurs, the algorithm must query all $n-k$
remaining intervals in the second stage. (Otherwise,
it could happen that all queries in the second stage
also yield a weight at the right end of the queried
interval, and then it is impossible to decide whether an unqueried
element is the minimum or not.)
So the expected cost of $ALG$ is
$$
\EX[ALG] = k + \frac{1}{2^k} (n-k) = \frac{n}{2^k} + k\left(1-\frac{1}{2^k}\right)
$$
If $k\ge \frac12 \log n$, the term $k(1-1/2^k)$ is $\Omega(\log n)$.
If $k < \frac12 \log n$, the term $n/2^k$ is $\Omega(\sqrt{n})$.
So in any case $\EX[ALG]=\Omega(\log n)$ while $\EX[\OPT]\le 2$.
\end{proof}

\subsection{The vertex split operation}
\label{app:vertexSplit}%
Lemma~\ref{lemma_opt_partitioning} allows us to partition
the vertex set $V$ of a hypergraph orientation instance
$H=(V,E)$
into parts, in such a way that
the expected optimal query cost can be lower-bounded by
the sum of the expected optimal query costs in each
parts. Each part inherits the probability for a
set of vertices to be the set of mandatory vertices in
the part from the original instance. Sometimes we
also need to consider partitions where a single
vertex can be split into fractions that are placed
in different parts of the partition.
In order to do so, we aim to create an instance $H'$ with query costs~$c'$ that is created from $H$ by splitting vertices into fractions and satisfies $\EX[\OPT'] = \EX[\OPT]$, where $\OPT'$ is the optimal query cost for instance $H'=(V',E')$.
In a first step, we consider the vertex cover instance $\bar{G}=(\bar{V},\bar{E})$.
Iteratively, we can replace the original vertices by copies:
We pick a vertex $v \in \bar{V}$ and replace it by multiple copies $v'_1,\ldots, v'_l$ with $c_v = \sum_{i=1}^l c'_{v'_i}$ such that there is an edge between a copy $v'_i$ and a vertex $u \in \bar{V} \setminus \{v\}$ if and only if $\{v,u\} \in \bar{E}$.
Clearly, the weight of the minimum vertex cover for the vertex cover instance remains the same.
Note that we cannot directly modify the vertex cover instance, instead we have to create hyperedges in $H'$ such that $H'$ has the described vertex cover instance.
For the described modifications, this is clearly possible.

For each $M \subseteq V$ let $M'$ be the set of copies created for vertices in $M$.
Let $p'(M')$ be the probability that $M'$ is the set of mandatory vertices for instance $H'$.
If we have $p(M) = p'(M')$ for each $M \subseteq V$, and $p'(S) = 0$ for each subset $S$ that is not the set of copies for some $M \subseteq V$, then it follows
\begin{align*}
\EX[\OPT] &= \sum_{M\subseteq V} p(M) \cdot c(M)  + \sum_{M\subseteq V}  p(M) \cdot c(VC_M)\\
&= \sum_{M\subseteq V'} p'(M) \cdot c'(M)  +  \sum_{M\subseteq V'}  p'(M) \cdot c'(VC'_{M}) = \EX[\OPT'],
\end{align*}
where $VC'_{M}$ is the minimum weight vertex cover for the subgraph  $\bar{G}'[V'\setminus M]$ of the vertex cover instance of $H'$.
In general, we cannot just set $p'(M)$ arbitrarily, since it depends on the structure of the given hypergraph and the distributions of the corresponding intervals.
However, since we use $\EX[\OPT']$ exclusively to lower bound $\EX[\OPT]$, we just consider the artificial generalized version of the graph orientation problem, where $p(M)$ for each $M \subseteq V$ is part of the input and can be freely set.
Then, we can create the described instance with $\EX[\OPT]=\EX[\OPT']$ and compare our algorithm against the optimal cost of the created instance for the generalized graph orientation problem.
Note that the construction of $p'$ implies $p_v = p'_{v'}$ where $v \in V$ is an original vertex and $v'$ is a copy of $v$ in $H'$.

\obsvertexsplitting*

\section{Appendix for Section~\ref{sec:threshold}}
\label{app:thres}


Figure~\ref{fig:alpha_ratio} shows the competitive ratio of $\threshold$ in dependence on the approximation factor $\alpha$ of the vertex cover problem blackbox. 
The following sections contain the missing proofs for Section~\ref{sec:threshold}.

\begin{figure}[b]
\centering
\begin{tikzpicture}
\begin{axis}[enlargelimits=false
			,ymin=1.5
			,width=10cm
			,height=5cm,
			,xlabel={Vertex cover approximation factor $\alpha$},
			ylabel={Competitive ratio}]
\addplot[domain=1:2, samples=100]{0.5 * (x+sqrt(8-x*(4-x)))};
\end{axis}
\end{tikzpicture}
\caption{Competitive ratio of \threshold for different approximation factors $\alpha$ of the vertex cover problem blackbox.}
\label{fig:alpha_ratio}
\end{figure}
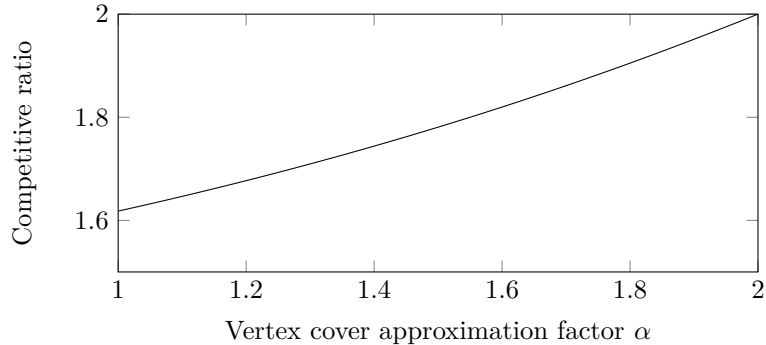

\subsection{Auxiliary lemma for the proof of Theorem~\ref{thm:threshold}}
\label{appThresholdBipartiteLemma}

\begin{restatable}[]{lemma}{lemmaLPPreprocessingBipartite}
	\label{lemma:LP_preprocessing_bipartite}
	Let $x^*$ be a half-integral optimal solution to the vertex cover LP relaxation, and let $V_1 = \{v \in V \mid x^*_v = 1\}$ and  $V_0 = \{v \in V \mid x^*_v = 0\}$.
	Then $V_1$ is a minimum-weight vertex cover for the bipartite graph $G'_{V_1\cup V_0}$ that is created from $G[V_1\cup V_0]$ by removing all edges within the partitions $V_1$ and $V_0$.
\end{restatable}

\begin{proof}
	We claim that $V_1$ is a minimum weighted vertex cover for $G'_{V_1\cup V_0}$.
	Assume otherwise; then there is a vertex cover $VC^*$ for $G'_{V_1 \cup V_0}$ with $c(VC^*) < c(V_1)$. 
	Let $R_1 = V_1 \setminus VC^*$ and let $A_0 = V_0 \cap VC^*$, then $VC^* = (V_1 \setminus R_1) \cup A_0$.
	Since $c(VC^*) < c(V_1)$, we get $c(R_1)>c(A_0)$.
	We define a solution $x'$ for the LP Relaxation of Step 4 as follows:
	$$
	x'_v = \begin{cases}
	x'_v = \frac{1}{2} & \text{if } v \in R_1 \cup A_0\\
	x'_v = x^*_v & \text{otherwise}.
	\end{cases}
	$$
	The objective value $c(x')$ of $x'$ is $c(x') = c(x^*) - \frac{1}{2} \cdot c(R_1) + \frac{1}{2} c(A_0)$, where $c(x^*)$ is the objective value of $x^*$.
	Since $c(R_1)>c(A_0)$, we get $c(x') < c(x^*)$.
	We argue that $x'$ is feasible for the LP relaxation, which contradicts the optimality of $x^*$.
	
	All edges that are only incident to elements of $V_{1/2} = V \setminus ({V_1 \cup V_0})$ are still covered, because the values of the corresponding variables were not changed.
	Each edge between some $u \in V_{1/2}$ and some $v \in V_1$ is still covered because $x'_u = \frac{1}{2}$ and $x'_v \ge \frac{1}{2}$ hold by definition of $x'$.
	Edges between elements of $V_1$ are covered for the same reason.
	Since there are no edges between elements of $V_0$, it remains to consider an edge between some $u \in V_0$ and some $v \in V_1$.
	If $v \in V_1 \cap VC^*$, then $x_v' = 1$, and the edge is covered.
	If $v \in V_1 \setminus VC^* = R_1$, then $u \in A_0$ follows because $VC^*$ is a vertex cover for $G'_{V_1 \cup V_0}$.
	By definition of $x'$, it follows $x'_v = \frac{1}{2}$ and $x_u' = \frac{1}{2}$, and, therefore, the edge is covered.
	It follows that $x'$ is feasible, which contradicts the optimality of $x^*$.
	Thus, $V_1$ is a minimum weighted vertex cover for $G'_{V_1 \cup V_0}$.
\end{proof}

\subsection{Tightness proof for the analysis of \threshold}
\label{appThreshold_tightness}

\begin{theorem}	
	The analysis of {\em \threshold} is tight.
	More precisely, there is no threshold $d\in (0,1]$ such that 
	its competitive ratio is less than the Golden ratio $\phi=(1+\sqrt{5})/2$.
\end{theorem}

\begin{proof}
	Consider the Threshold algorithm with some threshold $d\in (0,1]$. We give two instances of the graph orientation problem with uniform query cost and show that the algorithm has either a competitive ratio of at least $1+d-\epsilon$, for an arbitrarily small $\epsilon > 0$ (in particular, $\epsilon < d$), or at least $1/d$. This implies the lower bound.
	
	As the first instance, consider a single edge $\{u,v\}$ and mandatory probabilities $p_u=d-\epsilon$ and $p_v=0$.
	Threshold finds a basic feasible solution for the vertex cover LP relaxation for the edge $\{u,v\}$, say $x^{\textup{LP}}_u=0$ and $x^{\textup{LP}}_v=1$. It queries $v$, and then it still has to query $u$ with probability $p_u=d-\epsilon$, while it would have been sufficient to query only~$u$. The algorithm's competitive ratio is at least $1+d-\epsilon$.
	
	As a second instance, consider a star with the center $v$ and $n$ edges $\{v,u_i\}$ for $i\in\{1,2,\ldots,n\}$. Let $p_v=1$ and $p_{u_i}=d$. Then the algorithm queries all $v$ and $u_i$ due to the threshold $d$, whereas it would have been sufficient to query $v$ first and then continue with querying $u_i$ only if needed, that is, with probability $\EX[X_{u_i}]=d$ for each $i$. The algorithm has cost $n+1$, while the optimal cost is at most $1+n\cdot d$. Thus, the competitive ratio is at least $(n+1)/(1+n\cdot d)$, which tends to $1/d$ for $n\rightarrow \infty$.	
\end{proof}


\subsection{Proof of Theorem~\ref{thm:threshold} for arbitrary query costs}
\label{app:threshold_arbitrary_costs}

\threshold can also handle the graph orientation problem with arbitrary query costs.
In contrast to the uniform problem variant, the algorithm solves a weighted version of~\eqref{eq:threshold:LP}, and a weighted vertex cover problem in Line~\ref{line:threshold_solvevc}. 
In both instances, the query costs are used as weights.

\theoremThreshold*

\begin{proof}
We show that the analysis of Section~\ref{sec:threshold-algorithm} can be generalized for arbitrary query costs.
The expected cost of $\threshold$ is
\begin{align*}
\EX[ALG] &= c(Q) + \sum_{v\in V\setminus Q} c_v \cdot p_v\\
&= c(M) + c(V_1) + c(VC') + \sum_{v \in V_0} c_v \cdot p_v + \sum_{v \in V_{1/2} \setminus VC'} c_v \cdot p_v.
\end{align*}
Analogously to the analysis of \threshold for uniform query costs, we compare $\EX[ALG]$ and $\EX[\OPT]$ component-wise.
Remember that $\EX[\OPT]$ is characterized by Equation~\eqref{eq:LBonOPT}.
The bounds in Equations~\eqref{gt_eq_2} and~\eqref{gt_eq_5} regarding $\EX[\OPT_M]$ and $c(M)$, and $\EX[\OPT_{V_{1/2}}]$ and $c(VC') + \sum_{v \in V_{1/2} \setminus VC'} c_v \cdot p_v$, respectively, can be generalized straightforwardly using arguments analogous to those in Section~\ref{sec:threshold-algorithm}.
Generalizing Equation~\eqref{gt_eq_4} requires more effort.
	
Compare $c(S) + \sum_{v \in V_0} c_v \cdot p_v$ and $\EX[\OPT_{V_1 \cup V_0}]$. 
According to Lemma~\ref{lemma:LP_preprocessing_bipartite}, $V_1$ is a minimum weighted vertex cover for the bipartite graph $G'[V_1\cup V_0]$ that is created by removing all edges within the partitions $V_1$ and $V_0$ from the subgraph induced by $V_1 \cup V_0$.
In contrast to Section~\ref{sec:threshold-algorithm}, we cannot apply the  K\H{o}nig-Egerv\'{a}ry theorem for unweighted vertex covers.
Instead, the theorem gives us a function $\pi:E \rightarrow \mathbb{R}$ with $\sum_{\{u,v\} \in E} \pi(u,v) \le c_v$ for each $v \in V_1 \cup V_0$.
By duality theory, the constraint is tight for each $v \in V_1$, and $\pi(u,v) = 0$ holds if both $u$ and $v$ are in $V_1$. 
Thus, we can interpret $\pi$ as a function that distributes the complete weight $c_v$ of each $v \in V_1$ to the neighbors of $v$ outside of $V_1$.
For each $v \in V_0$ let $\tau_{v}$ denote the fraction of $c_v$ that is not used by $\pi$ to cover the weight of any $v\in V_1$, i.e., $\tau_v = c_v - \sum_{\{u, v\} \in E} \pi(u, v)$.
We can rewrite $c(V_1) + c(V_0)$ as
$$
	\sum_{v\in V_1} \left(c_v + \sum_{u\in V_0} \pi(u,v) \cdot p_v \right) + \sum_{v\in V_0} \tau_v \cdot p_v.
$$
Instead of comparing this expression to $\EX[\OPT_{V_1 \cup V_0}]$, we use Observation~\ref{obs:vertex_splitting} to split each $u \in V_0$ into copies $u'_v$ with query cost $c'_{u'_v} = \pi(u,v)$ for each $v \in V_1$ with $\pi(u,v) > 0$, and, if $\tau_u > 0$, in an additional copy $u'$ with query cost $c'_{u'} = \tau_u$. 
We assume that the mandatory probability of each copy is equal to the mandatory probability of the original.
Let $\OPT'_{V_1 \cup V'_0}$ denote random variable of Definition~\ref{def:partial_opt} for the modified instance, where $V'_0$ refers to the set of copies of vertices in $V_0$.
Then, according to Observation~\ref{obs:vertex_splitting}, $\EX[\OPT_{V_1 \cup V_0}] = \EX[\OPT'_{V_1 \cup V'_0}]$.
For each $v\in V_1$, let $H'_v = \{u'_v \mid u \in V_0 \land \pi(u,v) > 0\}$ be the set of copies that cover $c_v$ according to $\pi$.
Further let $\mathcal{H}' = \bigcup_{v\in V_1} H'_v$.
Using Lemma~\ref{lemma_opt_partitioning}, we further partition $\EX[\OPT'_{V_1 \cup V'_0}]$ and obtain
$$
	\EX[\OPT'_{S \cup D'}] \ge \sum_{v \in S}  \EX[\OPT'_{\{v\}\cup H'_v}] + \sum_{u' \in V'_0 \setminus \mathcal{H'}} \EX[\OPT'_{\{u'\}}].
$$
For each $s \in S$, the subgraph of the vertex cover instance induced by $\{v\}\cup H'_v$ is a star with a minimum weighted vertex cover of weight $c_v$ and thus 
$$
		\EX[\OPT'_{\{v\}\cup H'_v}] \ge c_v.
$$
Each $u' \in V'_0 \setminus \mathcal{H'}$ is a copy of some $u \in V_0$ with $\tau_u > 0$, and has cost $c'_{u'} = \tau_u$.
Since $\OPT'_{u'}$ must query $u'$ if it is mandatory, we have 
$$
		\EX[\OPT'_{V_1 \cup V'_0}] \ge c(V_1) + \sum_{u \in V'_0 \setminus \mathcal{H'}} \tau_u \cdot p_u.
$$
Finally, by putting the inequalities together, we get
\begin{align*}
		&c(S) + \sum_{v \in V_0} c_v \cdot p_v \\
		 \le &\sum_{v\in V_1} \left(c_v + \sum_{u\in V_0} \pi(u,v)\cdot d \right) + \sum_{v\in V_0} \tau_v \cdot p_v \\
		 \le &  \sum_{v\in V_1} (1+d) \cdot c_v + \sum_{v\in V_0} \tau_v \cdot p_v\\
		 \le &  \sum_{v\in V_1} (1+d) \cdot \EX[\OPT'_{\{v\}\cup H'_v}] + \sum_{u' \in V'_0 \setminus \mathcal{H'}} \tau_{u} \cdot p_u\\
		 \le& (1+d) \cdot \EX[\OPT'_{V_1 \cup V'_0}] = (1+d) \cdot \EX[\OPT_{V_1 \cup V_0}].
\end{align*}

The rest of the analysis follows the same pattern as Section~\ref{sec:threshold-algorithm}.
\end{proof}

\subsection{Applying the threshold algorithm to special cases}
\label{app:constant}

The threshold algorithm of Section~\ref{sec:threshold} requires an $\alpha$-approximation for a vertex cover problem.
In this section, we discuss further special cases for which \threshold can be applied with a polynomial running time. 
We show the following theorem.

\theoConstantSets*
We separately show the results for orienting $\mathcal{O}(\log \ |V|)$ hyperedges and sorting $\mathcal{O}(1)$ sets.

\paragraph*{Orienting $\mathcal{O}(\log \ |V|)$ hyperedges}  

Consider the problem of orienting a hypergraph $H=(V,E)$ with $|E| \le k$ for some $k \in \mathcal{O}(\log \ |V|)$.
In order to apply the threshold algorithm, we have to solve the minimum vertex cover problem for a subgraph of the vertex cover instance $\bar{G}$ of $H$.
Note that even in this special case, the number edges and vertices in $\bar{G}$ can be arbitrarily large.

\begin{theorem}
	Given a hypergraph orientation instance $H=(V,E)$ with $|E| \le k$ for some $k \in \mathcal{O}(\log \ |V|)$, a minimum weight vertex cover for $\bar{G}$ of $H$ can be computed in polynomial time. 
	Thus, for $\epsilon,\delta > 0$, \threshold is $1.618+\epsilon$-competitive with probability at least $1-\delta$.
\end{theorem}

\begin{proof}
	It only remains to show that the minimum weight vertex cover for $\bar{G}$ can be computed in polynomial time. Then, Theorem~\ref{theorem_threshold_hyper} implies the theorem.
	
	By definition of $\bar{G}$, for each hyperedge $F$, a vertex cover must contain either the left-most element $v_1$ of $F$ or all elements of $F\setminus \{v_1\}$.
	Thus, when computing a vertex cover, we, for each hyperedge $F$, have to decide whether to include $v_1$ or $F \setminus \{v_1\}$.
	The running time of enumerating through all such vertex covers is $\mathcal{O}(2^k)$, which is polynomial for $k \in \mathcal{O}(\log \ |V|)$.
\end{proof}

\paragraph*{Orienting the union of a constant number of interval graphs}
Consider the problem of orienting a graph $G=(V,E)$ with $E = \bigcup_{i=1}^k E_i$, where $G_i=(V_i,E_i)$ with $V_i \subseteq V$, for $1 \le i \le k$, is an interval graph.
Orienting such an $G_i$ corresponds to the problem of sorting a single set of uncertainty intervals.
Thus, orienting $G=(V,E)$ corresponds to the problem of sorting $k$ (possibly overlapping) sets.
In the following, we assume that $k$ is constant, and show the following theorem.

\begin{theorem}
	\label{theo_sorting_constant}
	A modified version of \threshold with polynomial running time is $1.618$-competitive for the problem of orienting a graph $G=(V,E)$ with $E = \bigcup_{i=1}^k E_i$, where $G_i=(V_i,E_i)$ with $V_i \subseteq V$, for $1 \le i \le k$, is an interval graph and $k$ is constant.
\end{theorem}

Since we assume preprocessed instances, i.e., $I_{v_1} \not\subseteq I_{v_2}$ holds for each edge $\{v_1,v_2\}$, it follows that each $G_i$ is a proper interval graph.

In order to use the threshold algorithm for this special case, we employ Algorithm~\ref{alg:clique_removal}, which uses a local ratio technique to remove all cliques $C$ of size $|C| \ge  3$ from $G$.

\begin{algorithm}[htb]
	\KwIn{Instance $G=(V,E)$ with query costs $c_v$ for each $v \in V$}
	Initialize $Q = \emptyset$\;
	\While{$G$ contains a clique $C$ with $|C| \ge 3$}{
		Let $\delta = \min_{v\in C} c_v$\;
		Set $c_v = c_v - \delta$ for each $v \in C$\label{Line_clique_reduction}\;
		Set $Q = Q \cup \{v \mid v\in C \land c_v = 0 \}$ and $V = V \setminus Q$\;		
	}
	\Return The modified $G$, the modified query costs $c_v$, and query set $Q$\;
	\caption{Clique removal}
	\label{alg:clique_removal}
\end{algorithm}

After applying this algorithm, each $G_i$ is a proper interval graph without cliques of size at least $3$, and, thus, triangle-free.

\begin{lemma}
	\label{lemma_vc_constant_width}
	Let $G=(V,E)$ be a graph orientation instance with $E = \bigcup_{i=1}^k E_i$, where $G_i=(V_i,E_i)$, for $1 \le i \le k$, with $V_i \subseteq V$ is a proper and triangle-free interval graph and $k$ is constant. Then, a minimum weight vertex cover of $G$ can be computed in polynomial time.
\end{lemma}

Before we prove the lemma, we use it to show that the clique removal allows us to solve the special case in polynomial time using the following algorithm:

\begin{enumerate}
	\item Execute Algorithm~\ref{alg:clique_removal} to obtain a modified graph $G$, modified query costs $c_v$, and a query set $Q_P$.
	\item Execute \threshold on the modified instance, add $Q_p$ to the query set $Q$ that is queried in Line~\ref{line:threshold_stage1} of \threshold.
	Use Lemma~\ref{lemma_vc_constant_width} as black box for the vertex cover problem.
\end{enumerate}

\begin{proof}[Proof of Theorem~\ref{theo_sorting_constant}]
	We show that this algorithm is $1.618$-competitive for the considered special case.
	Since Lemma~\ref{lemma_vc_constant_width} allows us to optimally solve the vertex cover problem,  we can apply Theorem~\ref{thm:threshold} with $\alpha=1$ to conclude that \threshold is $1.618$-competitive on the modified instance in Step~$2$ of the algorithm.
	
	Thus, it remains to argue about the cost for querying $Q_P$.
	Let $i$ be an iteration of the while-loop of Algorithm~\ref{alg:clique_removal}, let $C_i$ denote the clique that is considered in this iteration, and let $\delta_i$ denote the value $\delta$ of this iteration. 
	The cost for querying $Q_P$ is bounded by $c(Q_P) \le \sum_i |C_i| \cdot \delta_i$.
	Since $\OPT$ must query at least $|C_i|-1$ elements of each clique (cf.~Lemma~\ref{lemma:witness_set}), the reduction of the query costs in Line~\ref{Line_clique_reduction} reduces the cost of $\OPT$ by at least $(|C_i|-1)\cdot \delta_i$.
	Since this holds for each realization, $\EX[\OPT]$ is reduced by at least this amount as well.
	Let $\EX[\OPT_r]$ denote the expected reduction in the query cost of $\EX[\OPT]$ caused by the clique removal, then $\EX[\OPT_r] \ge \sum_i (|C_i|-1) \cdot \delta_i$.
	Thus, we can charge $c(Q_p)$ against $\EX[\OPT_r]$ to derive 
	$$
	\frac{c(Q_P)}{\EX[\OPT_r]} \le \max_i \left\{\frac{|C_i|}{(|C_i|-1)}\right\} \le \frac{3}{2},
	$$
	where the last inequality follows from $|C_i| \ge 3$.
	We can conclude that the algorithm is $\max\{1.5,\ 1.618\} = 1.618$-competitive.
\end{proof}

\paragraph*{Proof of Lemma~\ref{lemma_vc_constant_width}}

Let $G=(V,E)$ be a graph orientation instance with $E = \bigcup_{i=1}^k E_i$, where $G_i=(V_i,E_i)$, for $1 \le i \le k$, with $V_i \subseteq V$ is a proper and triangle-free interval graph and $k$ is constant.
Each $G_i$ being a proper and triangle-free interval graph implies that $G_i$ is the union of a set of disjoint paths.
Furthermore, since we assume preprocessed instances, i.e., $I_{v_1} \not\subseteq I_{v_1}$ for each edge $\{v_1,v_2\}$, the lower limits of the uncertainty intervals of each $G_i$ define a unique total order $\prec_i$ of the vertices in $V_i$.

We compute a minimum weight vertex cover via a dynamic program.
A \emph{state} $u$ of the dynamic program is a tuple $u=(P,S)$, where $P$ contains a single vertex $v_i$ for each $G_i$, and $S_{v_i}$ assigns each $v_i\in P$ to either $0$ or $1$.
The state $u$ represents a partial vertex cover for the sub graph induced by the vertices $v_i$ and all vertices that are before $v_i$ in the corresponding order $\prec_i$. 
The value $S_{v_i}$ denotes whether $v_i$ is part of the represented vertex cover ($S_{v_i}=1$) or not ($S_{v_i}=0$). 
We introduce two dummy vertices $s$ and $t$ with cost of zero, and add them to each $G_i$ as the first and last element of the order $\prec_i$.
Without loss of generality, we assume that $s$ and $t$ are part of every vertex cover.
We define $u_0 = (\{s\},1)$ as the initial state, representing that $s$ is part of every vertex cover.
Further, $u_f = (\{t\},1)$ represents a terminal state.
Since the number of subgraphs $G_i$ is bounded by $k$, the number of states is bounded by $\mathcal{O}(n^k \cdot 2^k)$.

We define a directed state graph $G' = (U,A,w)$, where $U$ is the set of DP states, such that the weight of the shortest $u_0$-$u_f$-path in $G'$ corresponds to the weight of a minimum weight vertex cover of $G$.
Consider a state $u=(P,S)$.
Let $v$ be a vertex, such that, for each $G_i$ with $v \in V_i$, the direct predecessor of $v$ according to $\prec_i$ is contained in $P$.
Note that, if $u\not= u_f$, such a vertex $v$ always exists.
Let $B \subseteq P$ denote the set of those predecessors, and 
let $B' = \{v' \in B \mid \{v',v\} \in E\}$ denote the subset of $B$ containing only vertices that share an edge with $v$.	
We add an arc $a$ between $u=(P,S)$ and $u'=(P',S')$ if the following holds:
\begin{enumerate}
	\item 
	If $S_{v'} = 1$ holds for all $v' \in B'$, we add an edge to $u'$ if $P' = \{ v' \in P \mid \exists 1\le i \le k: v' \in V_i \land v \not\in V_i\} \cup \{v\}$,
	 $S'_{v'} = S_{v'}$ for all $v' \in P \cap P'$, and $S'_{v} \in \{0,1\}$. 
	Intuitively, this arc models the extension of the partial vertex cover represented by $u$ by either adding $v$ to the partial vertex cover or deciding that $v$ is not part of the extended partial vertex cover.
	By definition, in the subproblem modeled by $u'$, the vertex $v$ only shares edges with elements of $B'$.
	In case $S_{v'} = 1$ holds for all $v' \in B'$, all elements of $B'$ are part of the vertex cover represented by $u$.
	Thus, in state $u'$, all vertices that share edges with $v$ are already part of the vertex cover, and vertex $v$ can therefore either be part of the represented vertex cover or not, in both cases the partial vertex cover is feasible.
	\item 
	If there exists an $v' \in B'$ with $S_{v'} = 0$, we additionally require $S'_{v}=1$. 
	In this case, $v$ shares an edge with a vertex $v'$ in the sub graph of $u'$, such that $v'$ is not part of the corresponding vertex cover. 
	Thus, $v$ must be part of the vertex cover.
\end{enumerate}
If $S'_v = 1$, we set $w_a = c_v$, which represents that the cost of the vertex cover increased by adding $v$.
Otherwise, we set $w_a = 0$.

A straightforward induction on the construction implies the following:
There is a vertex cover of $G$ with weight $C$ if and only if there is a $u_0$-$u_f$-path with cost $C$ in $G'$.
Since the number of states is bounded by $\mathcal{O}(n^k \cdot 2^k)$, we can find the shortest $u_0$-$u_f$-path with a running time exponential in $k$ but polynomial in the rest of the input.
This implies Lemma~\ref{lemma_vc_constant_width}.

\section{Appendix for Section~\ref{sec:VC-alg}}


\subsection{The graph orientation problem on a special weighted star}
\label{subsec:star:subproblem}

\begin{restatable}[]{lemma}{VertexCoverWeightedStar}
	\label{lemma:vertex_cover_weigthed_star}
	Consider the problem of orienting a star $G=(V,E)$  with center $v \in V$. 
	Let $L$ and $R$ be the vertex cover-based algorithms that query $VC_1 = \{v\}$ and $VC_2 = V\setminus\{v\}$ in the first stage, respectively.
	If $\EX[L] = \EX[R]$, then $L$ is $\frac{4}{3}$-competitive.	
\end{restatable}

\begin{proof}
The analysis relies on $p_v$ and $p_u$ being independent for each $u \in H_v$. 
In the graph orientation problem, for every vertex $u$ the probability $p_u$ corresponds to the probability that the exact value of at least one $v$ with $\{u,v\} \in E$ is contained in $I_u$. 
Since the graph is bipartite and there is an edge between $v$ and $u$, it follows that $v$ and $u$ do not share any adjacent vertices.
Thus, $p_v$ and $p_{u}$ are independent.
This also implies that $p_v$ and~$L$ are independent.
In the following let $H_v = V \setminus \{v\}$ denote the set of leaves. 
Thus, for the expected costs of the algorithms, we have
\begin{equation}
\EX[L] = c_v + \sum_{u\in V} p_u \cdot c_u = p_v \cdot c_v + c(H_v) = \EX[R].
\label{eq:LRdef}
\end{equation}

For each $v \in V$, let $X_v$ be an indicator variable denoting whether $v$ is mandatory.
Note that $p_v =\mathbb{P}[X_v=1]$.

\subparagraph{Case $c(H_v) \le c_v$.}
\label{subsubsec:star:assumption}

If $c(H_v) \le c_v$, querying the leaves is always the best strategy if the center is not mandatory, i.e., if $X_v = 0$, independently of whether the elements of $H_v$ are mandatory.
Similarly, if the center is mandatory, i.e., if $X_v = 1$, querying the center first is always the best strategy.
Thus, we can write $\EX[\OPT]$ as
\begin{align*}
	\EX[\OPT] &= p_v \cdot \EX[L \mid X_v = 1] + (1-p_v) \cdot \EX[R \mid X_v = 0]\\
	 &= p_v \cdot \EX[L] + (1-p_v) \cdot c(H_v)\\
	&= p_v \cdot \EX[L] + (1-p_v) \cdot (\EX[L] - p_v \cdot c_v)\\
	&\le \EX[L] - p_v \cdot (1-p_v) \cdot c_v,
\end{align*}
where the second equality uses the fact that $p_v$ and $L$ are independent, and the third equality uses $\EX[L] = \EX[R]$.
Since $0 \le p_v \le 1$, it follows that $p_v \cdot (1-p_v) \cdot c_v \le \frac{c_v}{4}$.
This directly implies $\frac{\EX[L]}{\EX[\OPT]} \le \frac{\EX[L]}{\EX[L] - p_v \cdot (1-p_v) \cdot c_v} \le \frac{4}{3}$, where the last inequality uses $\EX[L] \ge c_v$.
For the remainder of the analysis, we assume $c(H_v) > c_v$.

\subparagraph{Characterizing $\EX[L]$ and $\EX[\OPT]$.}

We begin by characterizing $\EX[L]$ in terms of conditional expectations:
$$
\EX[L] = p_v \cdot \EX[L \mid X_v = 1] + (1-p_v) \cdot \EX[L \mid X_v = 0].
$$
Since $X_v$ and $L$ are independent, it follows that $\EX[L \mid X_v = 1] = \EX[L]$.
If the center is mandatory, i.e., $X_v = 1$, we know that both $\OPT$ and $L$ have the same expected value.
That is, $\EX[\OPT \mid X_v = 1] = \EX[L \mid X_v = 1]$.

Thus, we are more interested in the case when $X_v = 0$.
For this case, we know that
$$\EX[L \mid X_v = 0] = \EX[L] =  c_v + \sum_{u\in H_v} p_u \cdot c_u = \EX[R] = p_v \cdot c_v + c(H_v).$$
But whether algorithm $L$ is the best possible strategy depends on how much query cost of the leaves is mandatory.
We have two cases.
In the first case, when $\sum_{u \in H_v} (1-X_u) \cdot  c_u \ge c_v$, querying the center first is a better strategy than querying the leaves first, even if the center is not mandatory.
In the second case, when $\sum_{u \in H_v} (1-X_u) \cdot c_u < c_v$, querying the leaves is the better strategy if additionally $X_v = 0$.

Consider the case when $\sum_{u \in H_v} (1-X_u) \cdot c_u \ge c_v$.
In this case, the cost of algorithm~$L$ is at most $c(H_v)$ and we can write $\EX[L \mid X_v = 0 \land c(H_v) \cdot (1-X_u) \ge c_v] = c(H_v) - s$, where $s \ge 0$ describes the slack between $c(H_v)$ and $\EX[L \mid X_v = 0 \land \sum_{u \in H_v} (1-X_u) \cdot c_u \ge c_v]$.
Define $q_L := \mathbb{P}\left[\sum_{u \in H_v} (1-X_u) \cdot c_u\ge c_v\right]$; then we have
\begin{align*}
\EX[L \mid X_v = 0] = q_L  \cdot \left(c(H_v) - s\right) + (1-q_L) \cdot \EX\Big[L \Big| X_v = 0 \land \sum_{u \in H_v} (1-X_u)  \cdot c_u< c_v\Big].
\end{align*}
Using this equation and $\EX[L \mid X_v = 0] = \EX[L] = \EX[R] = p_v \cdot c_v + c(H_v)$, we can derive
\begin{align*}
\EX\Big[L \Big| X_v = 0 \land \sum_{u \in H_v} (1-X_u) \cdot c_u < c_v\Big] = c(H_v) + \frac{1}{1-q_L} \cdot p_v \cdot c_v + \frac{q_L}{1-q_L} \cdot s.
\end{align*}
Finally, we use this to characterize $\EX[L]$:
\begin{align*}
	\EX[L] &= p_v \cdot \EX[L] + (1-p_v) \cdot q_L \cdot \left(c(H_v) - s\right)\\
	&+ (1-p_v) \cdot (1-q_L) \cdot \left(c(H_v) + \frac{1}{1-q_L} \cdot p_v \cdot c_v + \frac{q_L}{1-q_L} \cdot s\right)
\end{align*}

To characterize $\EX[\OPT]$, observe that the only case where $\OPT$ and $L$ deviate is the case where $X_v = 0$ and $\sum_{u \in H_v} (1-X_u) \cdot c_u < c_v$.
For that case, $\OPT$ queries only the leaves at cost $c(H_v)$.
Thus, we can write $\EX[\OPT]$ as
$$\EX[\OPT] = p_v \cdot \EX[L] + (1-p_v) \cdot q_L \cdot \left(c(H_v) - s\right) + (1-p_v) \cdot (1-q_L) \cdot c(H_v).$$

\subparagraph{Upper bounding the competitive ratio.}

We now use the characterizations of $\EX[L]$ and $\EX[\OPT]$ to show $\frac{4}{3}$-competitiveness.
Consider the difference between $\EX[L]$ and $\EX[\OPT]$,
$$
\EX[L] - \EX[\OPT] = (1-p_v) \cdot (p_v \cdot c_v + q_L \cdot s).
$$
We can use this difference to upper bound the competitive ratio:
$$
\frac{\EX[L]}{\EX[\OPT]} \le \frac{\EX[L]}{\EX[L]- (1-p_v) \cdot p_v \cdot c_v - (1-p_v) \cdot q_L \cdot s}.
$$

We continue by upper bounding $q_L \cdot s$.
Consider again the characterization of $\EX[L]$.
In the case when $X_v = 0$ and $\sum_{v\in H_v} (1-X_u) \cdot c_u < c_v$, the expected value of $L$ is larger than $c(H_v)$  by $\frac{1}{(1-q_L)} \cdot p_v \cdot c_v + \frac{q_L}{1-q_L} \cdot s$.
By Equation~(\ref{eq:LRdef}), the difference between $c(H_v)$ and the cost of $L$ is at most $c_v$.
It follows that
\begin{align*}
&\frac{1}{(1-q_L)} \cdot p_v \cdot c_v + \frac{q_L}{1-q_L} \cdot s \le c_v\\
\implies& p_v \cdot c_v + q_L \cdot s \le c_v \cdot (1 - q_L)\\
\implies& q_L \cdot s \le c_v \cdot (1-p_v) - c_v \cdot q_L \\
\implies& q_L \cdot (c_v + s) \le c_v \cdot (1-p_v)\\
\implies& q_L  \le c_v \cdot \frac{1-p_v}{c_v + s}
\end{align*}
Thus, we obtain 
$$
\frac{\EX[L]}{\EX[\OPT]} \le \frac{\EX[L]}{\EX[L]- (1-p_v) \cdot p_v \cdot c(v) - \frac{c_v \cdot (1-p_v)^2 \cdot s}{c_v+s}}.
$$
Moreover, we can observe that this upper bound decreases for increasing $\EX[L]$, so we only need a lower bound on $\EX[L]$.
In the characterization of $\EX[L]$, the expected cost of~$L$ for the case when $X_v = 0$ and $c(H_v) \cdot (1-p_u) \ge c_v$ is $c(H_v) - s$. 
Since the cost of $L$ can never be smaller than $c_v$, it follows that $s \le c(H_v) - c_v$, so the expected cost of the algorithm is $\EX[L] = \EX[R] = c(H_v) + c_v \cdot p_v \ge c_v + s + p_v \cdot c_v$.
By plugging this into the ratio we obtain:
\begin{align*}
\frac{\EX[L]}{\EX[\OPT]} &\le \frac{c_v + p_v \cdot c_v+ s}{c_v + p_v \cdot c_v + s -(1-p_v) \cdot p_v \cdot c_v - \frac{c_v \cdot (1-p_v)^2 \cdot s}{c_v+s}}\\
&= \frac{c_v + p_v \cdot c_v + s}{c_v + p_v \cdot c_v + s - (1 - p_v) \cdot p_v \cdot c_v - \frac{c_v \cdot (1-p_v)^2 \cdot s}{c_v+s}}
\end{align*}
We can now use  
$$\frac{c_v \cdot (1-p_v)^2 \cdot s}{c_v+s} = \frac{c_v \cdot (1-p_v)^2 \cdot s}{c_v \cdot (1+\frac{s}{c_v})} = \frac{(1-p_v)^2 \cdot s}{1+\frac{s}{c_v}}$$
to write the ratio as
\begin{align*}
\frac{\EX[L]}{\EX[\OPT]} &\le \frac{c_v + c_v \cdot p_v + s}{c_v + p_v \cdot c_v + s - (1 -p_v) \cdot p_v \cdot c_v - \frac{(1-p_v)^2 \cdot s}{1+\frac{s}{c_v}}}\\
&= \frac{1 + p_v + \frac{s}{c_v}}{1 + p_v + \frac{s}{c_v} - (1 - p_v) \cdot p_v - \frac{(1-p_v)^2 \cdot \frac{s}{c_v}}{1+\frac{s}{c_v}}}
\end{align*}
Assume $c_v > 0$; we can do this w.l.o.g.\ because we can query free elements in a preprocessing step. Let $s' = \frac{s}{c_v}$; then we have that 
\begin{align*}
\frac{\EX[L]}{\EX[\OPT]} 
&\le \frac{1 + p_v + s'}{1 + p_v + s' - (1 - p_v) \cdot p_v - \frac{(1-p_v)^2 \cdot s'}{1+s'}} = f(p_v, s').
\end{align*}

\subparagraph{Showing $\frac{4}{3}$-competitiveness.} To complete the proof, we show that $f(p_v, s') \leq \frac{4}{3}$ for any $0 \le p_v \le 1$ and $s' \ge 0$.
We can rewrite
\begin{displaymath}
 f(p_v, s') = \frac{(1 + s')(1 + p_v + s')}{1 + s' + (p_v + s')^2}.
\end{displaymath}
We fix the value of~$s'$ and look for critical points, so we consider the partial derivative of $f(p_v, s')$ on $p_v$, which is
\begin{displaymath}
 \frac{\partial}{\partial p_v} f(p_v, s') = - \frac{(1 + s')(p_v^2 + 2 p_v (1 + s') + s'^2 + s' - 1)}{((p_v + s')^2 + 1 + s')^2}.
\end{displaymath}
The denominator is clearly always greater than zero.

First let us consider $s' \geq \frac{1}{\phi} = \frac{\sqrt{5}-1}{2}$.
If $s'^2 + s' - 1 \geq 0$, which always holds for $s' \geq \frac{1}{\phi}$, then clearly $\displaystyle\frac{\partial}{\partial p_v} f(p_v, s') \leq 0$, so $f(p_v, s')$ is non-increasing and is maximized when $p_v = 0$.
Let $g(s') = f(0, s') = \displaystyle\frac{(1+s')^2}{1 + s' + s'^2}$; then $\displaystyle\frac{d}{ds'} g(s') = \displaystyle\frac{1-s'^2}{(1+s'+s'^2)^2}$, and the only critical point is obtained by taking $\displaystyle \frac{d}{ds'} g(s') = 0$, which holds when $s' = 1$.
It is clear that $\displaystyle \frac{d}{ds'} g(s') \geq 0$ for $0 \leq s' \leq 1$ and $\displaystyle \frac{d}{ds'} g(s') \leq 0$ for $s' \geq 1$, so $g(s')$ has a global maximum when $s' = 1$.
Thus the maximum value of $f(p_v, s')$ for $s' \geq \frac{1}{\phi}$ and $p_v \geq 0$ is $f(0, 1) = \frac{4}{3}$.

Now assume $0 \leq s' < \frac{1}{\phi}$.
Let us look at the critical points by taking $\displaystyle\frac{\partial}{\partial p_v} f(p_v, s') = 0$.
Since the denominator is always positive and $s' \geq 0$, we only have a zero when
\begin{displaymath}
  N(p_v, s') = p_v^2 + 2 p_v (1 + s') + s'^2 + s' - 1 = 0, 
\end{displaymath}
which by the quadratic formula yields $p_v = -1 -s' \pm \sqrt{s' + 2}$.
Since we need $p_v \geq 0$, we have $p_v = -1 -s' + \sqrt{s' + 2}$, and $0 \leq p_v \leq 1$ since $0 \leq s' < \frac{1}{\phi}$.
We claim that this is always a maximum point.
Clearly $\displaystyle\frac{\partial}{\partial p_v} f(p_v, s') \geq 0$ whenever $N(p_v, s') \leq 0$ and vice-versa, but $\displaystyle\frac{\partial}{\partial p_v} N(p_v, s') = 2(1 + p_v + s')$, which is non-negative for $p_v, s' \geq 0$, so $N(p_v, s')$ is non-decreasing.
We can thus conclude that $\displaystyle\frac{\partial}{\partial p_v} f(p_v, s') \geq 0$ for $0 \leq p_v \leq -1 -s' + \sqrt{s' + 2}$ and $\displaystyle\frac{\partial}{\partial p_v} f(p_v, s') \leq 0$ for $-1 -s' + \sqrt{s' + 2} \leq p_v \leq 1$, so we have a maximum point when $p_v = -1 - s' + \sqrt{s' + 2}$.
Finally, let 
\begin{displaymath}
 h(s') = f(-1 -s' + \sqrt{s'+2}, s') = \frac{(1 + s')\sqrt{s' + 2}}{2s' + 4 -2 \sqrt{s' + 2}};
\end{displaymath}
we have that $\displaystyle\frac{d}{ds'} h(s') = \frac{1}{4 \sqrt{s' + 2}}$, which is always non-negative for $s' \geq 0$.
Thus $h(s')$ is a non-decreasing function, so its maximum for $0 \leq s' < \frac{1}{\phi}$ is attained when $s'$ tends to $\frac{1}{\phi}$.
Therefore, for $0 \leq s' < \frac{1}{\phi}$ and $p_v \geq 0$, we have that $f(p_v, s') \leq h(\frac{1}{\phi}) = \frac{\sqrt{5}+3}{4} < \frac{4}{3}$.
\end{proof}


\subsection{Proof of Theorem~\ref{thm:VC-single-set}}
\label{subsec:proof-vc-single-set}

\theoremVCSingleSet*

\begin{proof}
Let $S= \{v_0,v_1,\ldots,v_n\}$ be the single hyperedge.
We assume that $v_0,v_1,\ldots,v_n$ are ordered by non-decreasing 
left endpoints of the corresponding intervals $I_0,\ldots,I_n$.
That is, $v_0$ is the leftmost vertex of $S$ and,
by our assumptions in Section~\ref{sec:preliminaries}, 
all intervals $I_1,\ldots,I_n$ intersect $I_0$, 
but none of them is contained in $I_0$.

In this setting, the vertex cover instance is a star with center $v_0$
and leaves $v_1,\ldots,v_n$.
Thus, we only have to consider two vertex cover-based algorithms:
the algorithm $L$ that queries vertex cover $\{v_0\}$ in the first stage, 
and the algorithm $R$ that queries $\{v_1,\ldots,v_n\}$ in the first stage.
Note that, while the vertex cover instance is a star, we cannot apply the analysis of Appendix~\ref{subsec:star:subproblem}:
This is because that analysis requires the mandatory probabilities of $v_0$ and each $v_i$ with $1 \le i \le n$ to be independent, which is not the case here.

If $n\ge 3$, it is sufficient to consider algorithm $L$.
The algorithm simply queries the vertices in order of
left endpoints, starting with $v_0$, until we can identify
the minimum element. 
If $v_0$ is also queried by $\OPT$,
we have $L=\OPT$. Otherwise, $\OPT$ must query all
of $v_1,\ldots,v_n$, and the competitive ratio is at most
$\frac{n+1}{n} \le \frac{4}{3}$.
Since $L$ is $\frac{4}{3}$-competitive in this case, so is \bestVC.

So it remains to consider the cases $n=1$ and $n=2$.
For $n=1$, the configuration of the intervals is shown in Figure~\ref{fig:min1set2}.
Here, $p$ is the probability that the weight of $v_0$ is contained in~$I_1$, and $q$ is the probability that the weight of $v_1$ is contained in $I_0$. 
Note that in this case $p$ and $q$ correspond to the mandatory probabilities of $v_1$ and $v_0$.

\begin{figure}[h]
  \centering
  \begin{tikzpicture}[line width = 0.3mm, scale=1.2]
    \interval{$I_0$}{0}{2}{0}
    \interval{$I_1$}{1}{3}{-0.5}

    \draw[dotted] (1, 0.5) -- (1, -1);
    \draw[dotted] (2, 0.5) -- (2, -1);

    \node at (0.5, 0.2) {$1-p$};
    \node at (1.5, 0.18) {$p$};

    \node at (1.5, -0.32) {$q$};
    \node at (2.5, -0.3) {$1-q$};
  \end{tikzpicture}
  \caption{Configuration for a hyperedge with two vertices.}
  \label{fig:min1set2}
\end{figure}
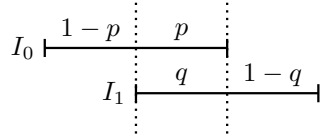

We can assume that $q\ge p$ (otherwise, we swap $I_0$ and $I_1$ and flip the x-axis). Since $\OPT$ has to query at least one vertex and only queries a second one if both are mandatory, we have $\E[\OPT]=1+pq$.
Algorithm $L$ queries $I_0$ first and has $\E[L]=(1-p)\cdot 1+ p \cdot2 = 1+p$. The ratio $\displaystyle\frac{1+p}{1+pq}$ for $0\le p\le 1$ and $p\le q\le 1$ is maximized for $q=p$, in which case it is $\displaystyle\frac{1+p}{1+p^2}$.
The derivative of $f(p)=\displaystyle\frac{1+p}{1+p^2}$ is $f'(p)=\displaystyle\frac{1-2p-p^2}{(1+p^2)^2}$, which is positive for $0\le p< -1+\sqrt{2}$, equal to $0$ for $p=-1+\sqrt{2}$, and negative for $1+\sqrt{2}<p\le 1$.
Hence, the maximum of $f(p)$ in the range $0\le p\le 1$ is attained at $p_0=-1+\sqrt{2}\approx 0.4142$
with value $f(p_0)=\frac{1+\sqrt{2}}{2}\approx 1.207$.
Thus, 
\bestVC is $1.207$-competitive for $n=1$.

It remains to consider case $n=2$. The configuration of the intervals is shown in Figure~\ref{fig:min1set3}.
Note that the order of the right endpoints of $I_1$ and $I_2$ is irrelevant for the proof.

\begin{figure}[h]
  \centering
  \begin{tikzpicture}[line width = 0.3mm, scale=1.2]
    \interval{$I_0$}{0}{3}{0}
    \interval{$I_1$}{1}{4}{-0.5}
    \interval{$I_2$}{2}{5}{-1}

    \draw[dotted] (1, 0.5) -- (1, -1.5);
    \draw[dotted] (2, 0.5) -- (2, -1.5);
    \draw[dotted] (3, 0.5) -- (3, -1.5);

    \node at (0.5, 0.2) {$p_0$};
    \node at (1.5, 0.2) {$p_1$};
    \node at (2.5, 0.2) {$p_2$};

    \node at (1.5, -0.3) {$q_1$};
    \node at (2.5, -0.3) {$q_2$};
    \node at (3.5, -0.3) {$q_3$};

    \node at (2.5, -0.8) {$r_2$};
    \node at (4, -0.8) {$r_3$};
  \end{tikzpicture}
  \caption{Configuration for three intervals.}
  \label{fig:min1set3}
\end{figure}
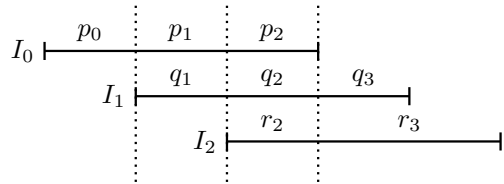

Again, we only have to consider the two vertex cover-based algorithms $L$ and $R$. 
We have:
\begin{eqnarray*}
\E[\OPT] &=& p_0 \cdot 1 + p_1 \cdot 2 + p_2 (q_1 \cdot 2 + q_2\cdot 3+q_3(r_2\cdot 3+r_3 \cdot 2))\\
\E[L] &=& p_0\cdot 1 + p_1\cdot 2 + p_2(q_1\cdot2+q_2\cdot3+q_3\cdot3)\\
\E[R] & = & 3-q_3 r_3
\end{eqnarray*}
By definition, 
\bestVC has an expected value of at most $\min\{ \E[L], \E[R] \}$, and, thus, the competitive ratio is at most
$$
\frac{\min\{ \E[L], \E[R] \}}{\E[\OPT]}.
$$

First, we want to remove some parameters to simplify our calculations.
Note that if we move $\varepsilon$ probability from $q_2$ to $q_1$,
then $\E[\OPT]$ and $\E[L]$ decrease by $p_2\varepsilon$, while $\E[R]$ is unchanged. So we can assume that $q_2=0$.
Furthermore, if we move $\varepsilon$ probability fro $p_1$ to~$p_0$,
then $\E[\OPT]$ and $\E[L]$ decrease by $\varepsilon$, while $\E[R]$ is unchanged. So we can assume that $p_1=0$.
This gives:
\begin{eqnarray*}
\E[\OPT] &=& p_0 \cdot 1 + p_2 (q_1 \cdot 2 + q_3(r_2\cdot 3+r_3 \cdot 2))\\
\E[L] &=& p_0\cdot 1 + p_2(q_1\cdot2+q_3\cdot3)\\
\E[R] & = & 3-q_3 r_3
\end{eqnarray*}
Now we rename the parameters as
\begin{displaymath}
p := p_2, \quad 1 - p = p_0, \quad q := q_1, \quad 1-q = q_3, \quad r := r_2, \quad 1-r = r_3.
\end{displaymath}
The equations then become:
\begin{eqnarray*}
\E[\OPT] &=& 1+p+p(1-q)r\\
\E[L] &=& 1+p(2-q)\\
\E[R] & = & 2+r+q-qr
\end{eqnarray*}

The corresponding picture is shown in Figure~\ref{fig:min1set3renamed}.

\begin{figure}[h]
  \centering
  \begin{tikzpicture}[line width = 0.3mm, scale=1.2]
    \interval{$I_0$}{0}{3}{0}
    \interval{$I_1$}{1}{4}{-0.5}
    \interval{$I_2$}{2}{5}{-1}

    \draw[dotted] (1, 0.5) -- (1, -1.5);
    \draw[dotted] (2, 0.5) -- (2, -1.5);
    \draw[dotted] (3, 0.5) -- (3, -1.5);

    \node at (0.5, 0.2) {$1-p$};
    \node at (1.5, 0.2) {$0$};
    \node at (2.5, 0.2) {$p$};

    \node at (1.5, -0.3) {$q$};
    \node at (2.5, -0.3) {$0$};
    \node at (3.5, -0.3) {$1-q$};

    \node at (2.5, -0.8) {$r$};
    \node at (4, -0.8) {$1-r$};
  \end{tikzpicture}
  \caption{Configuration for three intervals assuming $p_1 = q_2 = 0$.}
  \label{fig:min1set3renamed}
\end{figure}
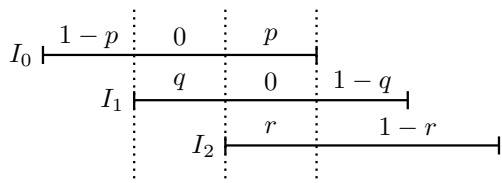

\paragraph*{Case 1: $\E[L] \le \E[R]$.}
We want to show:
\begin{equation}
\E[L] \le \frac43 \cdot \E[\OPT]
\label{case1claim}
\end{equation}
This inequality (\ref{case1claim}) can be transformed as follows:
\begin{eqnarray*}
3+3p(2-q) &\le& 4+4p+4p(1-q)r \\
\Leftrightarrow
2p+4pqr & \le &1+4pr+3pq
\end{eqnarray*}
This is clearly satisfied for $p=0$, so we can divide by $p$ and get:
\begin{eqnarray}
2+4qr & \le & \frac1p +4r+3q
\label{case1claimB}
\end{eqnarray}
With respect to $p$, this inequality is most difficult to satisfy
if $p$ is as large as possible. To see how large $p$ can be, we
look back at our assumption for the current case,
$\E[L] \le \E[R]$.
Expanding this assumption gives:
$$
1+p(2-q) \le 2+r+q-qr
$$
This can be reformulated as:
\begin{equation}
p \le \frac{1+r+q-qr}{2-q} \label{condp}
\end{equation}
We now distinguish two cases depending on whether the right-hand side of (\ref{condp}) is larger than~$1$ or not.

\paragraph*{Case 1.1: $\displaystyle\frac{1+r+q-qr}{2-q}\ge 1$.}
This condition is equivalent to:
\begin{eqnarray}
1+r+q-qr & \ge & 2-q \nonumber \\
\Leftrightarrow r+2q & \ge & 1+qr \label{cond11}
\end{eqnarray}
In this case, we can set $p$ to $1$. Instead of inequality (\ref{case1claimB}), it now suffices to
show the following:
\begin{eqnarray}
2+4qr & \le & 1 +4r+3q \nonumber \\
\Leftrightarrow
1+4qr & \le & 4r+3q \label{case1claimC}
\end{eqnarray}
By (\ref{cond11}) we have:
$$
1+4qr = 1+qr + 3qr \le r+2q+3qr
$$
Thus, to show (\ref{case1claimC}) it suffices to show:
\begin{eqnarray*}
r+2q+3qr & \le & 4r+3q \\
\Leftrightarrow 3qr & \le & 3r+q
\end{eqnarray*}
This is obviously satisfied as $3qr\le 3r$ follows from $q\le 1$.

\paragraph*{Case 1.2: $\displaystyle\frac{1+r+q-qr}{2-q}< 1$.}
This condition is equivalent to:
\begin{eqnarray}
1+r+q-qr & < & 2-q \nonumber \\
\Leftrightarrow
r+2q & < & 1+qr \nonumber \\
\Rightarrow
4qr+4q^2r^2 & > & 4r^2q + 8 rq^2 \label{cond12}
\end{eqnarray}
(In the last step, we have multiplied the inequality with $4qr$.)
In this case, we can set $p$ to $\displaystyle\frac{1+r+q-qr}{2-q}$.
Instead of inequality (\ref{case1claimB}) it now suffices to show the following:
\begin{eqnarray*}
2+4qr & \le & \frac{2-q}{1+r+q-qr} +4r+3q \\
\Leftrightarrow
(1+r+q-qr)(2+4qr-4r-3q) & \le & 2-q \\
\Leftrightarrow
2-2r-q-5qr-4r^2-3q^2+7q^2r+8qr^2-4q^2r^2
 & \le & 2-q \\
\Leftrightarrow
7q^2r+8qr^2 & \le & 2r+qr+4r^2+3q^2+(4q^2r^2+4qr)
\end{eqnarray*}
By (\ref{cond12}), it suffices to show that:
$$
7q^2r+8qr^2  \le  2r+qr+4r^2+3q^2+4r^2q+8rq^2
$$
This can be transformed into:
$$
4qr^2 \le 2r+qr+4r^2+3q^2+rq^2
$$
This clearly holds because $4qr^2\le 4r^2$ follows from $q\le 1$.

\paragraph*{Case 2: $\E[L] \ge \E[R]$.}
We want to show:
\begin{equation}
\E[R] \le \frac43 \cdot \E[\OPT]
\label{case2claim}
\end{equation}
This inequality (\ref{case2claim}) can be transformed as follows:
\begin{eqnarray}
6+3r+3q-3qr &\le& 4+4p+4p(1-q)r \nonumber \\
\Leftrightarrow
2+3r+3q+4pqr & \le & 3qr+4p+4pr \nonumber \\
\Leftrightarrow
p(4+4r-4qr) & \ge & 2+3r+3q-3qr
\label{case2claimB}
\end{eqnarray}
Expanding our assumption
$\E[L] \ge \E[R]$
gives:
$$
2+r+q-qr \le 1+p(2-q)
$$
This can be reformulated as:
\begin{eqnarray}
1+r+q+pq & \le & 2p+qr \nonumber \\
\Leftrightarrow
p(2-q) & \ge & 1+r+q-qr \nonumber \\
\Leftrightarrow
p & \ge & \frac{1+r+q-qr}{2-q}
\label{cond2}
\end{eqnarray}
We observe that this is only possible if the right-hand side of (\ref{cond2}) is at most~$1$, so we
must have the following for this to be possible, but we will not need to use this fact in the
remainder of the proof.
\begin{eqnarray}
1+r+q-qr &\le& 2-q \nonumber\\
r+2q &\le& 1+qr \nonumber
\end{eqnarray}
Inequality (\ref{case2claimB}) is most difficult to satisfy if $p$ is as
small as possible, so in view of (\ref{cond2}) we can set $p=\displaystyle\frac{1+r+q-qr}{2-q}$
and then prove the following inequality to establish~(\ref{case2claimB}):
\begin{eqnarray}
\frac{1+r+q-qr}{2-q}(4+4r-4qr) & \ge & 2+3r+3q-3qr \nonumber \\
\Leftrightarrow
2r+5qr+4r^2+3q^2+4q^2r^2 & \ge & 8qr^2+7q^2r \label{case2claimC}
\end{eqnarray}
We show that (\ref{case2claimC}) holds by showing the following two inequalities:
\begin{eqnarray}
5qr+4r^2 &\ge& 8qr^2 \label{case2claimD1}\\
2r+3q^2+4q^2r^2 &\ge&7q^2r \label{case2claimD2}
\end{eqnarray}
Inequality (\ref{case2claimD1}) holds because $5qr\ge 5qr^2$ and $4r^2\ge 4qr^2$, so
$5qr+4r^2\ge 9qr^2\ge 8qr^2$.
To show (\ref{case2claimD2}), we distinguish two cases for~$r$.

\paragraph*{Case 2.1: $r\le \frac35$.}
In this case $3q^2 \ge 5 q^2r$ and hence
$2r+3q^2 \ge 2q^2r + 5 q^2r = 7q^2r$, and thus (\ref{case2claimD2}) holds.

\paragraph*{Case 2.2: $r > \frac35$.}
In this case we have $4q^2r^2\ge 4q^2r \cdot \frac{3}{5}\ge 2 q^2r$,
and hence $2r+3q^2+4q^2r^2 \ge 2q^2r+3q^2r+2q^2r=7q^2r$, so (\ref{case2claimD2}) also
holds in this case.
\end{proof}


\subsection{Proof of Theorem~\ref{thm:lb1setuniform}}

\theoremLBSingleSetUniform*

\begin{proof}
Consider $n+1$ elements $e_0, \ldots, e_n$, with $I_0 = (0, 2)$ and $I_i = (1, 3)$ for $i > 0$, and uniform query costs $c_i = 1$ for all $i$.
The probabilities are such that $\mathbb{P}[v_0 \in (1, 2)] = \frac{n-1}{n}$ and $\mathbb{P}[v_i \in (1, 2)] = \epsilon$ for all $i > 0$, with some $0 < \epsilon \ll \frac{n-1}{n}$.
It is easy to show that $\EX[\OPT] = n - \frac{n-1}{n}(1-\epsilon)^n$, which tends to $\frac{n^2 - n + 1}{n}$ as $\epsilon$ approaches $0$.
By the argumentation in~\cite[Section~3]{chaplick20stochasticLATIN}, the best decision tree either queries $e_0$ and then all other elements if $v_0 \in (1, 2)$, or it queries $e_1, \ldots, e_n$ adaptively (in an order that does not matter for this instance since they are all identical) until either a query is forced in $e_0$ (when some $v_i \in (0, 1)$ with $i > 0$) or the problem is solved by querying all elements except $e_0$ (when $v_i \in [1, 2)$ for all $i > 0$).
The first decision tree clearly has expected query cost $1 + n \cdot \mathbb{P}[v_0 \in (1, 2)] = n$.
The expected query cost of the second decision tree can be computed with the following recurrence, where $A(k)$ is the cost of the subtree for elements $e_0$ and $e_{n-k+1}, \ldots, e_n$:
\begin{displaymath}
 \left\{ \begin{array}{l} A(1) = 1 + \epsilon \\ A(k) = 1 + \epsilon \cdot \left( 1 + \frac{n-1}{n} \cdot (k -1) \right) + (1-\epsilon) \cdot A(k-1) \end{array} \right.
\end{displaymath}
It is easy to show by induction that $A(k) = k - \frac{k}{n} + \left( 1 + \frac{1}{n\epsilon}\right)(1 - (1-\epsilon)^k)$, so we have that
$A(n) = n - (1-\epsilon)^n + \frac{1 - (1-\epsilon)^n}{n \epsilon}$ and therefore
\begin{displaymath}
 \lim_{\epsilon \rightarrow 0} A(n) = n - 1 + \lim_{\epsilon \rightarrow 0} \frac{1 - (1-\epsilon)^n}{n \epsilon} =  n - 1 + \lim_{\epsilon \rightarrow 0} \frac{n(1-\epsilon)^{n-1}}{n} = n - 1 + 1  = n,
\end{displaymath}
where the second equality follows by L'Hôpital's rule.
For either decision tree, the competitive ratio tends to ${n^2}/({n^2-n+1})$ as $\epsilon$ approaches $0$.
\end{proof}

\end{document}